\documentclass[9pt,reqno]{amsart}  
\calclayout

\usepackage[english]{babel}
\usepackage[p,osf]{cochineal}
\usepackage[scale=.95,type1]{cabin}
\usepackage[zerostyle=c,scaled=.94]{newtxtt}
\usepackage[T1]{fontenc} 
\usepackage[utf8]{inputenc} 
\usepackage{amsthm}
\usepackage{amsmath}
\usepackage{amssymb}
\usepackage{pgfplots} 
\pgfplotsset{compat=newest} 
\usepgfplotslibrary{fillbetween,colorbrewer}
\usepackage{amsfonts}
\usepackage{amsaddr}
\usepackage{xspace,subcaption}
\usepackage{enumitem} 
\usepackage{csquotes}
\usepackage{xcolor}
\usepackage[colorlinks]{hyperref}
\hypersetup{colorlinks, breaklinks, citecolor=teal,filecolor=black}
\usepackage{graphicx}
\usepackage{mathtools}
\mathtoolsset{centercolon}
\usepackage{bm}
\usepackage{pgfkeys}
\usepackage{pgf,interval}
\intervalconfig{soft open fences}

\DeclareMathOperator{\Tr}{Tr}

\DeclareMathOperator{\supp}{supp}

\DeclareMathOperator{\Cov}{Cov}

\DeclareMathOperator{\E}{\mathbf{E}}
\DeclareMathOperator{\Prob}{\mathbf{P}}

\newcommand{\ov}{\overline}
\newcommand{\ii}{\mathrm{i}}

\ifdefined\C
\renewcommand{\C}{\mathbf{C}}
\else
\newcommand{\C}{\mathbf{C}}
\fi

\newcommand{\un}{\underline}
\newcommand{\vx}{\bm{x}}

\newcommand{\R}{\mathbf{R}}
\newcommand{\N}{\mathbf{N}}
\newcommand{\Z}{\mathbf{Z}}

\newcommand{\cO}{\mathcal{O}}
\newcommand{\co}{{\scriptstyle\mathcal{O}}}

\newcommand{\dif}{\operatorname{d}\!{}}
\DeclarePairedDelimiter{\braket}{\langle}{\rangle}%
\DeclarePairedDelimiter{\abs}{\lvert}{\rvert}%
\DeclarePairedDelimiter{\norm}{\lVert}{\rVert}%
\providecommand\given{}
\newcommand\SetSymbol[1][]{\nonscript\:#1\vert\allowbreak\nonscript\:\mathopen{}}
\DeclarePairedDelimiterX{\tuple}[1](){\renewcommand\given{\SetSymbol[\delimsize]}#1}
\DeclarePairedDelimiterX{\set}[1]\{\}{\renewcommand\given{\SetSymbol[\delimsize]}#1}
\DeclarePairedDelimiterXPP{\landauO}[1]{\cO}(){}{#1}
\DeclarePairedDelimiterXPP{\landauo}[1]{\co}(){}{#1}
\DeclarePairedDelimiterXPP{\landauOprec}[1]{\cO_\prec}(){}{#1}

\numberwithin{equation}{section} 
\newtheorem{theorem}{Theorem}[section]
\newtheorem{assumption}{Assumption}
\newtheorem{lemma}[theorem]{Lemma}
\newtheorem{proposition}[theorem]{Proposition}
\newtheorem{remark}[theorem]{Remark}
\newtheorem{corollary}[theorem]{Corollary}

\date{\today}
\author{Giorgio Cipolloni \and L\'aszl\'o Erd\H{o}s}
\address{IST Austria, Am Campus 1, 3400 Klosterneuburg, Austria}
\author{Dominik Schr\"oder\(^{\ast}\)}
\address{Institute for Theoretical Studies, ETH Zurich, Clausiusstr.\ 47, 8092 Zurich, Switzerland}
\email{giorgio.cipolloni@ist.ac.at} 
\email{lerdos@ist.ac.at}
\email{dschroeder@ethz.ch}
\thanks{\(^\ast\)Supported by Dr.\ Max R\"ossler, the Walter Haefner Foundation and the ETH Z\"urich Foundation}
\title{Quenched universality for deformed Wigner matrices}
\date{\today}
\begin{document}
\thispagestyle{empty}

\begin{abstract}
   Following E.\ Wigner's original vision, we prove that sampling  the eigenvalue gaps
   within the bulk spectrum of  a \emph{fixed} (deformed) Wigner matrix $H$ yields the celebrated Wigner-Dyson-Mehta
   universal statistics with high probability.
   Similarly, we prove universality for
   a monoparametric family of deformed Wigner matrices $H+xA$ with
   a deterministic Hermitian matrix $A$ and a fixed Wigner matrix \(H\), just using the randomness of a single scalar real random variable $x$.
   Both results constitute \emph{quenched} versions
   of bulk universality that has so far only been proven in \emph{annealed} sense
   with respect to the probability space of the matrix ensemble.
\end{abstract}
\maketitle

\section{Introduction}
Random matrix theory in physics was originally envisioned  by  E.\ Wigner
to predict statistics of gaps between the energy levels of heavy atomic nuclei. The
underlying physical systems have no inherent disorder and the statistical ensemble
in Wigner's description was generated by randomly (uniformly) sampling from the experimentally
measured gaps of a fixed nucleus
within a large energy range. The model ensemble, the space of Hermitian random matrices
with independent, identically distributed entries (Wigner ensemble), is however inherently random.
Accepting the replacement of the original physical Hamiltonian with a Hermitian random matrix, one may ask
whether uniform sampling within the  spectrum  of a \emph{fixed, typical  realisation} of a Wigner matrix also gives rise to the celebrated Wigner-Dyson-Mehta  (WDM) universality. In this paper  we affirmatively answer this question, in the sense that for any fixed typical Wigner matrix the empirical gap statistic is close to the Wigner surmise, see Figure~\ref{figure quenched}. We thus prove a stronger version of WDM universality
and confirm the applicability of Wigner's theory  even in the  \emph{quenched sense}.
All previous universality proofs, see e.g.~\cite{MR1727234,1712.03881,MR3405746, MR3687212, MR3502606, MR3541852, MR4134946, MR4026551, MR1810949,MR2810797,MR2784665,MR3699468,MR3372074} (see also~\cite{MR3351052,MR3253704,MR3192527,MR3390602,MR2306224,MR2012268} for invariant ensembles),
were valid in the \emph{annealed sense}, i.e.\ where the eigenvalue statistics were directly generated
by the randomness of the matrix ensemble.

\begin{figure}
   \centering
   \[\set[\Big]{100\rho(\lambda_{50}(H))[\lambda_{51}(H)-\lambda_{50}(H)] \given H\sim \mathrm{Wig}_{100}}\]
   \begin{tikzpicture}
      \begin{axis}[width=12cm,height=2.5cm,
            ymin=0,xmax=3,ymax=1.2, ytick={0,0.5,1}, xtick={0,1,2,3},
            axis lines=center,
         ]
         \addplot[ybar,bar width=.075,fill=black!10,bar shift=0.0] table[col sep=comma,x index=0,y index=1] {quenched.csv};
         \addplot [thick, domain=0:3, samples=201] {32*x^2/pi^2*exp(-4*x^2/pi)};
      \end{axis}
   \end{tikzpicture}
   \[\set[\Big]{N\rho(\lambda_{k}(H))[\lambda_{k+1}(H)-\lambda_{k}(H)] \given k=\frac{N}{10}, \frac{N}{10}+1,
         \ldots, \frac{9N}{10} }\]
   \begin{tikzpicture}
      \begin{axis}[width=12cm,height=2.5cm,
            ymin=0,xmax=3,ymax=1.2, ytick={0,0.5,1}, xtick={0,1,2,3},
            axis lines=center,
         ]
         \addplot[ybar,bar width=.075,fill=black!40,bar shift=0.0] table[col sep=comma,x index=0,y index=2] {quenched.csv};
         \addplot [thick, domain=0:3, samples=201] {32*x^2/pi^2*exp(-4*x^2/pi)};
      \end{axis}
   \end{tikzpicture}
   \caption{The two types of universality: The first histogram shows the normalised gaps of the two middle eigenvalues in the spectrum of \(5000\) complex Wigner matrices of size \(100\times 100\). The second histogram shows the empirical normalised bulk eigenvalue gaps of a single complex Wigner matrix of size \(5000\times 5000\). Both distributions asymptotically approach the \emph{Gaudin-Mehta distribution} \(p_2\) drawn as solid lines, see Section~\ref{sec painleve}.}\label{figure quenched}
\end{figure}
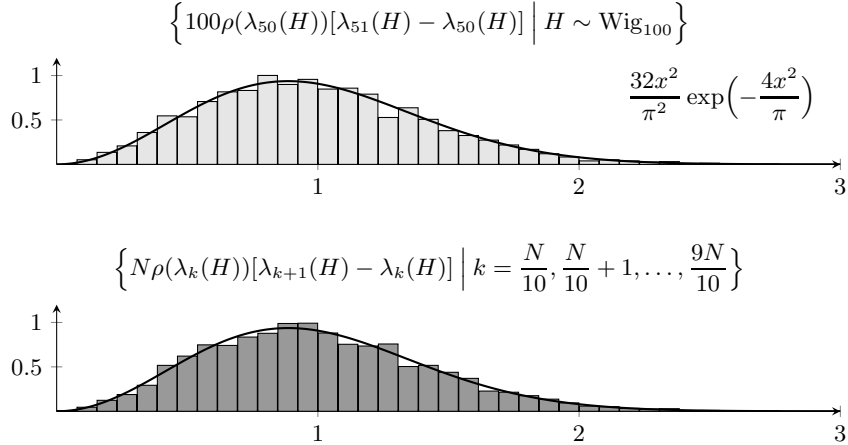

More generally, we consider random matrices of the form \(H^x:=H+xA\), where
\(H\) is a large  \(N\times N\)  Wigner matrix, \(A\) is a fixed nontrivial  Hermitian deterministic matrix,
and \(x\) is a real random variable
(in fact we can even consider more general \emph{deformed Wigner matrices \(H\)}).
We show that  for a typical  but fixed (quenched) \(H\) the randomness of \(x\) alone is
sufficient to generate WDM universality in the bulk of the limiting spectrum of \(H^x\), i.e.\ we prove that the local statistics of \(H+xA\) are universal for \emph{all fixed}  \(H\) in a high probability set.  The special case \(A=I\) and
\(x\) being uniformly distributed on some small interval yields Wigner's spectral sampling model.
Another special case covered by our general result is when \(A\) itself is chosen from a Wigner ensemble. The corresponding \(H+xA\) model for the Gaussian case was introduced by
H. Gharibyan, C. Pattison, S. Shenker\footnote{Private communication  via Stephen Shenker and Sourav Chatterjee in June 2020.} and  K. Wells who coined it as
the  \emph{monoparametric ensemble}~\cite{GPSW}.

The basic guiding principle for establishing quenched universality of \(H^x\) is to show that
near a fixed energy \(E\) the eigenvalues of \(H^x\) and \(H^{x'}\)
are essentially uncorrelated whenever \(x\) and \(x'\) are not too close. This
provides the sufficient (asymptotic) independence along the sampling in the space of \(x\).
Following a similar idea in~\cite{1912.04100} for a different setup,  the
independence of eigenvalues is proven by running the Dyson Brownian motion (DBM)
for the matrix \(H\). The corresponding stochastic differential equations for the
eigenvalues of \(H^x\) and \(H^{x'}\) have almost independent stochastic differentials
if the corresponding eigenvectors are asymptotically orthogonal. Therefore,
independence of eigenvalues can be achieved by running the DBM already after
a short time, provided we can understand eigenvector overlaps. The small
Gaussian component added along the DBM flow can later be removed by fairly
standard perturbation argument \emph{(Green function comparison theorem)}.

Thus the main task is to show that eigenvectors of \(H^x\) become asymptotically orthogonal for different, sufficiently distant values of \(x\).
This orthogonality can be triggered by two quite different mechanisms
that we now explain.

The first mechanism is present
when \(A\) is not too close to a diagonal matrix, in other words if \(\mathring{A}:= A- \braket{A}\)
is nontrivial in the sense that
\(\braket{\mathring{A}^2}\ge c\) with some \(N\)-independent constant \(c>0\).
Here \(\braket{A}: =\frac{1}{N}\Tr A\) denotes the normalized trace.
In this case the entire eigenbasis of \(H^x\) is \emph{rotated}, i.e.\ it becomes essentially  orthogonal  to
that of \(H^{x'}\) whenever \(x\) and \(x'\) are not too close. As a consequence, the entire spectra of \(H^x\) and \(H^{x'}\) are essentially uncorrelated.
To establish this effect of \emph{eigenbasis rotation},
we use a \emph{multi-resolvent local law} for the resolvents of \(H^x\) and \(H^{x'}\);
this method currently requires \(|x-x'|\ge N^{-\epsilon}\) for typical choices of \(x, x'\).
To  ensure this, we assume that  \(x=N^{-a}\chi \) where \(\chi\) is an \(N\)-independent real random variable
with some regularity and \(a\in[0,\epsilon]\). %

The  second mechanism is the most transparent when \(A=I\) and
\(x=N^{-a}\chi \) where \(\chi\)
is uniformly distributed on some small fixed interval;
we call this mechanism the \emph{sampling in the spectrum}.  In this case the eigenbasis of \(H^x\) actually does not depend on \(x\). However, the eigenvectors corresponding to eigenvalues close to a fixed energy are algebraically orthogonal for sufficiently distant \(x,x'\). We also prove that distant eigenvalue gaps of \(H\),
and hence the \emph{local spectral data}
of \(H^x,H^{x'}\) are essentially uncorrelated.

By the \emph{rigidity} property of the eigenvalues, already a small change in \(x\)  triggers this effect, so it works in
the entire range of scales \(a\in [0, 1-\epsilon]\). Moreover,  the proof can easily be extended to
more complicated random matrix ensembles well beyond the Wigner case.
No multi-resolvent local law is needed in the proof.

A combination of these two mechanisms can  be used in the situation when
\(A\ne I\), but \(A\) is still close to \(\braket{A}\) times the identity in
the sense that \(|\braket{A}|\ge C\braket{\mathring{A}^2}^{1/2}\) for some large \(C\). This extension   complements
the main condition \(\braket{\mathring{A}^2}\ge c\) needed in the first mechanism
thus proving the result unconditionally for any \(A\).

\subsection*{Notations and conventions}
We introduce some notations we use throughout the paper. For integers \(k\in\N \) we use the notation \([k]:= \set{1,\ldots, k}\). For positive quantities \(f,g\) we write \(f\lesssim g\) and \(f\sim g\) if \(f \le C g\) or \(c g\le f\le Cg\), respectively, for some constants \(c,C>0\) which depend only on the \emph{model parameters} appearing in our base Assumptions~\ref{ass:x}--\ref{ass:entr}. For any two positive real numbers $\omega_*, \omega^*\in\mathbf{R}_+$ by $\omega_*\ll \omega^*$ we denote that $\omega_*\le c\omega^*$ for some small constant $0<c<1/100$. We denote vectors by bold-faced lower case Roman letters \({\bm x}, {\bm y}\in\C ^k\), for some \(k\in\N\). Vector and matrix norms, \(\norm{\vx}\) and \(\norm{A}\), indicate the usual Euclidean norm and the corresponding induced matrix norm.  For vectors \({\bm x}, {\bm y}\in\C^k\) we define
\[
   \braket{ {\bm x},{\bm y}}:= \sum_i \overline{x}_i y_i
\]
and for any \(N\times N\) matrix \(A\) we use the notation \(\braket{ A}:= N^{-1}\Tr  A\) to denote the normalized trace
of \(A\). We will use the concept of ``with very high probability'' meaning that for any fixed \(D>0\) the probability of an \(N\)-dependent event is bigger than \(1-N^{-D}\) if \(N\ge N_0(D)\). Moreover, we use the convention that \(\xi>0\) denotes an arbitrary small constant which is independent of \(N\).

\subsection*{Acknowledgement.}  The  authors are indebted to Sourav Chatterjee for forwarding  the very inspiring
question that Stephen Shenker originally addressed to him which  initiated the current paper.
They are also grateful  that
the authors of~\cite{GPSW} kindly shared  their  preliminary numerical results in June 2021.

\subsection*{Data availability.}
All data generated or analysed during this study are included in this manuscript.

\section{Main results}
In this paper we consider real and complex \emph{Wigner matrices}, i.e.\ Hermitian \(N\times N\) random matrices \(H=H^\ast\) with independent identically distributed (i.i.d.)\ entries (up to Hermitian symmetry)
\begin{equation}\label{eq wigner mat}
   h_{ab}\stackrel{d}{=} N^{-1/2}\begin{cases}
      \chi_\mathrm{od}, & a<b, \\
      \chi_\mathrm{d},  & a=b,
   \end{cases} \qquad h_{ba}:=\ov{h_{ab}}
\end{equation}
having finite moments of all orders, i.e. \(\E\abs{\chi_\mathrm{od}}^p+\E\abs{\chi_\mathrm{d}}^p\le C_p\). The entries are normalised such that \(\E\abs{\chi_\mathrm{od}}^2=1\), and additionally \(\E\chi_\mathrm{od}^2=0\) in the complex case. The normalisation guarantees that the eigenvalues \(\lambda_1\le \lambda_2\le\ldots \le\lambda_N\)
of \(H\) asymptotically follow Wigner's semicircular distribution \(\rho_\mathrm{sc}(x):=\sqrt{(4-x^2)_+}/(2\pi)\).
In the \emph{bulk regime}, i.e.\ where  \(\rho_\mathrm{sc}\ge c\) for some \(c>0\), the
eigenvalue gaps are of order; \(\lambda_{i+1}-\lambda_i\sim 1/N\).

The \emph{Wigner-Dyson-Mehta conjecture} for the bulk of Wigner matrices \(H\) asserts that for any \(i\in[\epsilon N,(1-\epsilon)N]\) the distribution of the rescaled eigenvalue gap %
converges%
\begin{equation}\label{wdm gap}
   \lim_{N\to\infty}\Prob\Bigl(N \rho_\mathrm{sc}(\lambda_i)[\lambda_{i+1}-\lambda_i]\le y\Bigr) = \int_0^y p_\beta(t)\dif t
\end{equation}
to a universal distribution with density \(p_1\) (for real symmetric Wigner matrices) or \(p_2\) (for complex Hermitian Wigner matrices) which can be computed explicitly from the integrable Gaussian GOE/GUE ensembles, see~Section~\ref{sec painleve} later.
This WDM conjecture was resolved in~\cite{MR3372074} while
similar results with a small averaging in the index \(i\) were proven earlier~\cite{MR2810797, MR2784665}.

As a corollary to our main result Theorem~\ref{thm2b} below on the \emph{monoparametric ensemble} we prove a considerable strengthening of~\eqref{wdm gap}, namely that with high probability the \emph{sampling}
of eigenvalues within a single fixed Wigner matrix generates WDM universality.
\begin{corollary}[to Theorem~\ref{thm2b}]\label{thm2}
   Let \(H\) be a Wigner matrix and \(I\subset (-2+\epsilon,2-\epsilon)\) be an interval in the bulk of \(H\) of length \(\abs{I}\ge N^{-1+\xi}\) for some \(\epsilon,\xi>0\). Then there exist small \(\kappa, \alpha>0\) and an event \(\Omega_I\) in the probability space \(\Prob_H\) of \(H\) with \(\Prob_H(\Omega_I^c)\le N^{-\kappa}\), such that for all \(H\in \Omega_I\) it holds that
   \begin{equation}\label{univ2}
      \sup_{y\ge 0}\abs*{
         \frac{1}{N\int_I\rho_\mathrm{sc} }\#\set*{i\given N\rho_\mathrm{sc}(\lambda_i)[\lambda_{i+1}-\lambda_i]\le y,\, \lambda_i\in  I}-\int_0^y p_\beta(t)\dif t}= \landauO{N^{-\alpha}},
   \end{equation}
   where the implicit constant in~\eqref{univ2} and \(\kappa,\alpha\) depend on \(\epsilon, \xi\).
\end{corollary}

Our main results are on the \emph{quenched (bulk) universality} of \emph{monoparametric random matrices}
\begin{equation}
   H^x := H + x A
\end{equation}
for general deterministic Hermitian matrices \(A\) of the same symmetry class\footnote{This restriction apparently excludes
   the case when \(A\) is complex Hermitian but \(H\) is real symmetric. With a slight modification of our proof (similar to the modification required in~\cite[Section 7]{MR4235475} compared to~\cite[Section 7]{1912.04100}), however, we can handle this
   case as well, but for brevity we refrain from presenting it.} as \(H\), and independent scalar random variables \(x\), just using the randomness of \(x\) for any fixed Wigner matrix \(H\) from a high probability set. For \(A=I\) the monoparametric universality of \(H^x\) implies the spectral sampling universality as stated in Corollary~\ref{thm2}, see Section~\ref{appendix prop}. Our results extend beyond Wigner matrices, we also allow for arbitrary additive deformations (certain results even extend to Wigner matrices with correlated entries), and cover general sufficiently regularly distributed scalar random variables \(x\).

\begin{assumption}[Deformed Wigner matrix]\label{ass:entr}
   We consider deformed Wigner matrices of the form \(H=W+B\), where \(W\) is a Wigner matrix as in~\eqref{eq wigner mat}, and \(B=B^\ast\) is an arbitrary deterministic matrix of bounded norm, i.e.\ \(\norm{B}\le C_0\) for some \(N\)-independent constant \(C_0\).

\end{assumption}
\begin{assumption}\label{ass:x}
   Assume that
   \(x=N^{-a}\chi\) with \(a\in [0,1)\), where \(\chi\) is an \(N\)-independent compactly supported real random variable such that for any small \(b_1>0\) there exists \(b_2>0\) such that for any interval \(I\subset\R\) with \(|I|\sim N^{-b_1}\) it holds \(\Prob(\chi\in I)\le |I|^{b_2}\).
\end{assumption}

To state the result, we now introduce the self-consistent density of states of \(H^x=W+B+xA\).
It has been proven in~\cite[Theorem 2.7]{MR3916109} that the resolvent \(G^x(z)= (H^x-z)^{-1}\) of \(H^x\) at a spectral parameter \(z\in \C\setminus\R\) can be well approximated
by the unique deterministic matrix \(M=M^x(z)\), solving the
\emph{Matrix Dyson Equation} (MDE)~\cite{MR4164728} (see also~\cite{MR2376207})
\begin{equation}\label{MDE}
   -M^{-1}=z-B - xA+\braket{M}, \qquad \Im M(z) \Im z> 0.
\end{equation}
We define the \emph{self consistent density of states} (scDos)~\cite[Section 4.1]{MR4164728} of \(H^x\) as
\begin{equation}\label{scdos}
   \rho^x(E):=\lim_{\eta\to0^+}\frac{1}{\pi}\braket{\Im M^x(E+\ii\eta)},
\end{equation}
and, in particular, the scDos of \(H\) by \(\rho:=\rho^0\). It is well known that \(\rho^x\) is a probability density which is compactly supported and real analytic inside its support~\cite[Proposition~2.3]{MR4164728}.
For the special case \(\E H=B=0\) the scDos of \(H\) is the standard Wigner semicircle law, i.e.\ \(\rho=\rho_\mathrm{sc}\).

We say that an energy \(E\in \R\) lies in the \emph{bulk of the spectrum} of \(H^x\) if \(\rho^x(E)\ge c\) for some \(N\)-independent constant \(c>0\).
For \(E\) in the bulk, the solution \(M^x(z)\) can be continuously extended to the real line, \(M^x(E):=\lim_{\eta\to 0^+}M^x(E+\ii\eta)\), and \(M^x(E+\ii\eta)\) for \(E\) in the bulk is \emph{uniformly bounded}, cf.~\cite[Proposition 3.5]{MR4164728}.
Finally, we define the \emph{classical eigenvalue locations}  to  be the \emph{quantiles} of \(\rho^x\), i.e.\ we define \(\gamma_i^x\) by
\begin{equation}\label{eq:defquantl}
   \int_{-\infty}^{\gamma_i^x} \rho^x(\tau)\dif \tau=\frac{i}{N},\qquad i\in [N].
\end{equation}

For clarity, in this section we only present single-gap versions of both mechanisms explained
in the introduction that yield quenched universality. Subsequently we will present the multi-gap analogues in Section~\ref{sec multi gap}.%
\subsection{Monoparametric universality via eigenbasis rotation}\label{mono univ rot}
The main universality result for the first mechanism (\emph{eigenbasis rotation}) is the following quenched
fixed-index universality result for the \emph{monoparametric ensemble}. We denote the probability measure and expectation of \(x\) by \(\Prob_x,\E_x\) in order to differentiate it from the probability measure \(\Prob_H\) of \(H\).

\begin{theorem}[Quenched universality for monoparametric ensemble]\label{thm1}
   Let \(H\) be a deformed Wigner matrix satisfying Assumption~\ref{ass:entr}, and let \(x=N^{-a}\chi\) be a scalar real random variable satisfying Assumption~\ref{ass:x} with \(a\in [0,a_0]\), where \(a_0\) is a small universal constant\footnote{Following the explicit constants along the proof,
      one may choose \(a_0=1/100\)}. Fix any \(c_0, c_1>0\)  small constants and assume that \(\braket{\mathring{A}^2}\ge c_0\), with \(\mathring{A}:=A-\braket{A}\). Suppose that \(i\in [N]\) is a \emph{bulk index}\footnote{To specify the \(c_1\)-dependence,
      we often speak of \(c_1\)-bulk index.}
   for \(H^x=H+xA\),  i.e.\ it holds that
   \begin{equation}\label{scond}
      \rho^x(\gamma^x_i)\ge c_1 \quad\text{for \(\Prob_x\)-almost all \(x\).}
   \end{equation}
   Then there exist small
   \(\alpha, \kappa>0\) and an event \(\Omega_i=\Omega_{i, A}\) with \(\Prob_H(\Omega_i^c)\le N^{-\kappa}\),
   so that for all \(H\in \Omega_i\) the  statistics of the  \(i\)-th rescaled gap of the eigenvalues \(\lambda_i^x\) of \(H^x\)  is universal, i.e.
   \begin{equation}\label{univ1}
      \abs*{\E_x f\Big(  N \rho^x(\lambda^x_{i})[ \lambda^x_{i+1}-  \lambda^x_{i}] \Big) - \int_0^\infty  p_\beta(t)
      f(t)\dif t}=\mathcal{O}(N^{-\alpha}\norm{f}_{C^5})
   \end{equation}
   for any smooth, compactly supported function \(f\) where the implicit constant in~\eqref{univ1} depends on \(a_0,c_0,c_1\) and the diameter of $\supp f$, and \(\kappa, \alpha\) depend on \(a_0\).
\end{theorem}
\begin{remark}\label{remark}We mention a few simple observations about Theorem~\ref{thm1}.
   \begin{enumerate}[label=(\roman*)]
      \item By the regularity of \(f\), \(\rho^x\) and  by rigidity of the bulk eigenvalues (see~\eqref{eq:rig} later) we may replace the random scaling factor \(\rho^x(  \lambda^x_{i}  )\) with  \(\rho^x(  \gamma^x_{i}  )\) at negligible error.
      \item For \(\E H=0\) and \(a>0\) the condition~\eqref{scond} can simply be replaced
            by \(i\in [N\epsilon', N(1-\epsilon')]\) for some \(\epsilon'>0\)
            and the argument of \(f\) in~\eqref{univ1} simplifies to
            \(f\big(  N \rho_{\mathrm{sc}}(  \lambda^x_{i} )[ \lambda^x_{i+1}-  \lambda^x_{i}] \big)\).
      \item Empirically we find that the convergence towards the universal gap statistics in~\eqref{univ1} is much slower for the monoparametric ensemble compared to GUE, cf.\ Figure~\ref{figure quenched2}. While even for \(2\times 2\) GUE matrices the empirical gap distribution is already very close to the Gaudin-Mehta distribution (see Section~\ref{sec painleve}), we observe the same effect only for large monoparametric matrices.
   \end{enumerate}
\end{remark}

\begin{remark}\label{rem3}
   We mention an interesting special case of Theorem~\ref{thm1} when \(H\) is a Wigner matrix and \(A\) itself is chosen from a Wigner ensemble that is independent of \(H\) and \(x\). In this case Theorem~\ref{thm1} implies\footnote{The condition on \(\braket{\mathring{A}^2}\) is satisfied  since \(\braket{A^2} = 1+ o(1)\) and \(\braket{A}=o(1)\) with  very high probability. Moreover, the scDos \(\rho^x\) is very close to a rescaled semicircle law with radius \(2\sqrt{1+x^2}\) with very high probability in the joint probability space of \(H\) and \(A\), hence the condition~\eqref{scond} holds for all \(i\in [N\epsilon', N(1-\epsilon')]\) for some \(\epsilon'>0\).} that for any fixed pair of Wigner matrices \(A,H\) from a high probability set, the universality of the \(i\)-th gap statistics of \(H+xA\) for \(i\in [N\epsilon', N(1-\epsilon')]\) is solely
   generated by the single real random variable \(x\), i.e.
   \begin{equation}
      \Prob_{H, A} \Bigl( \text{\(i\)-th gap statistics of \(H+xA\) is universal} \Bigr) = 1- \mathcal{O}(N^{-\epsilon}).
   \end{equation}
   This mathematically rigorously answers to a question of Gharibyan,
   Pattison, Shenker and  Wells~\cite{GPSW}.
   While their original question referred to a standard Gaussian \(x\), which is not compactly supported, a simple cut-off argument extends our proof to this case as well.
\end{remark}
\begin{figure}
   \centering
   \begin{tikzpicture}
      \begin{axis}[width=5cm,height=3cm,
            ymin=0,xmax=3.2,ymax=1.2, ytick={.5,1}, xmajorticks=false,title={\(N=2\)},ylabel={GUE}
         ]
         \addplot[ybar,bar width=.2,fill=black!20,bar shift=0.0,draw=none] table[col sep=comma,x index=0,y index=1] {GUEvsMono2.csv};
         \addplot[thick] table[col sep=comma,x index=1,y index=0] {surmise.csv};
      \end{axis}
   \end{tikzpicture}
   \begin{tikzpicture}
      \begin{axis}[width=5cm,height=3cm,
            ymin=0,xmax=3.2,ymax=1.2, ymajorticks=false , xmajorticks=false,title={\(N=100\)}
         ]
         \addplot[ybar,bar width=.2,fill=black!20,bar shift=0.0,draw=none] table[col sep=comma,x index=0,y index=2] {GUEvsMono2.csv};
         \addplot[thick] table[col sep=comma,x index=1,y index=0] {surmise.csv};
      \end{axis}
   \end{tikzpicture}
   \begin{tikzpicture}
      \begin{axis}[width=5cm,height=3cm,
            ymin=0,xmax=3.2,ymax=1.2, ymajorticks=false , xmajorticks=false,title={\(N=1000\)}
         ]
         \addplot[ybar,bar width=.2,fill=black!20,bar shift=0.0,draw=none] table[col sep=comma,x index=0,y index=3] {GUEvsMono2.csv};
         \addplot[thick] table[col sep=comma,x index=1,y index=0] {surmise.csv};
      \end{axis}
   \end{tikzpicture}
   \begin{tikzpicture}
      \begin{axis}[width=5cm,height=3cm,
            ymin=0,xmax=3.2,ymax=1.2, ytick={.5,1}, xtick={0,1,2}, ylabel={monop.}
         ]
         \addplot[ybar,bar width=.2,fill=black!40,bar shift=0.0,draw=none] table[col sep=comma,x index=0,y index=4] {GUEvsMono2.csv};
         \addplot[thick] table[col sep=comma,x index=1,y index=0] {surmise.csv};
      \end{axis}
   \end{tikzpicture}
   \begin{tikzpicture}
      \begin{axis}[width=5cm,height=3cm,
            ymin=0,xmax=3.2,ymax=1.2, ymajorticks=false , xtick={0,1,2},
         ]
         \addplot[ybar,bar width=.2,fill=black!40,bar shift=0.0,draw=none] table[col sep=comma,x index=0,y index=5] {GUEvsMono2.csv};
         \addplot[thick] table[col sep=comma,x index=1,y index=0] {surmise.csv};
      \end{axis}
   \end{tikzpicture}
   \begin{tikzpicture}
      \begin{axis}[width=5cm,height=3cm,
            ymin=0,xmax=3.2,ymax=1.2, ymajorticks=false, xtick={0,1,2},
         ]
         \addplot[ybar,bar width=.2,fill=black!40,bar shift=0.0,draw=none] table[col sep=comma,x index=0,y index=6] {GUEvsMono2.csv};
         \addplot[thick] table[col sep=comma,x index=1,y index=0] {surmise.csv};
      \end{axis}
   \end{tikzpicture}
   \caption{Here we show the histogram of a (rescaled) single eigenvalue gap
      \(\lambda_{N/2+1}-\lambda_{N/2}\)  in the middle of the spectrum for \(N\in\set{2,100,1000}\) and for random matrices sampled from either the GUE or the monoparametric ensemble. For the GUE ensemble the histogram has been generated by sampling \(2000\) independent GUE matrices \(H\). For the monoparametric ensemble only two GUE matrices \(H,A\) have been drawn at random, and the histogram has been generated by sampling \(2000\) standard Gaussian random variables \(x\) and considering the gaps of  \(H^x=H+xA\). The solid black line represents the theoretical limit \(p_2(s)\) which matches the empirical distribution very closely already for \(N=2\). In Appendix~\ref{appendix figure} we present numerical evidence for the speed of convergence, inspired by the observation on the slow convergence of the spectral form factor made in~\cite{GPSW}.}\label{figure quenched2}
\end{figure}
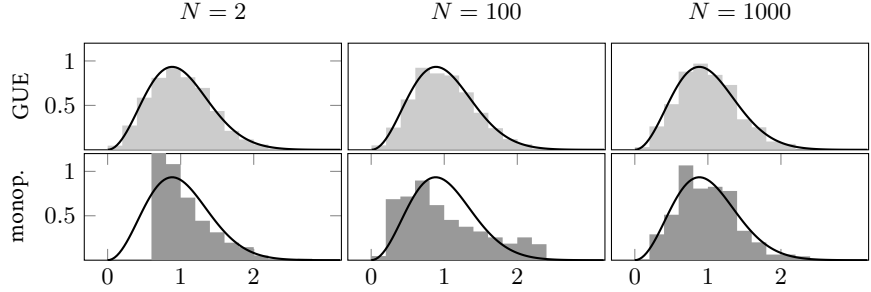
\subsection{Monoparametric universality via spectral sampling}\label{mono univ samp}
The main universality result for the second mechanism (\emph{spectral sampling}) is the following quenched
fixed-energy universality result for the \emph{monoparametric ensemble}. We define \(i_0=i_0(x,E)\) as the index such that \(\gamma_{i_0}^x\) is the quantile of \(\rho^x\) closest to \(E\), i.e.
\begin{equation}
   i_0(x,E) := \biggl\lceil N\int_{-\infty}^E \rho^x(t) \dif t \biggr\rceil \in [N]
\end{equation}
with \(\lceil\cdot\rceil\) denoting rounding to the next largest integer.

For the special case \(A=I\) (formulated as~\ref{case1} in Theorem~\ref{thm2b} below)
we obtain quenched sampling universality for a much broader class of Hermitian random matrices \(H\) with slow correlation decay\footnote{\label{cor mat def}These are \(N\times N\) Hermitian matrices \(H\) with
\emph{covariance operator} \(\mathcal{S}[R] := \frac{1}{N} \E WRW\), where \(W:=\sqrt{N}(H-\E H)\)
is a correlated centred random matrix. Note that this \(W\) is \(\sqrt{N}\)-times bigger
than the Wigner matrix \(W\) defined in Assumption~\ref{ass:entr}. This notational inconsistency
occurs only in this description of the correlated ensemble where we follow the
convention of~\cite{MR3941370}.
We assume that \(\norm{\E H}\le C\) and
that \(W\) satisfies Assumptions (B)--(E) of~\cite{MR3941370}. We recall that Assumption (B) requires that  all moments of the matrix elements of \(W\) are finite, i.e.\ \(\E |W_{ab}|^q\le C_q\) with some constant \(C_q\) for any \(q\) integer, uniformly in the indices \(a, b\in [N]^2\), while Assumption (E) requires  that the covariance operator satisfies the so called \emph{flatness condition}
\[
   c\braket{R}\le \mathcal{S}[R] \le C\braket{R}
\]
for any positive semi-definite matrix \(R\), where
\(c, C\) are some  fixed positive constants.
Finally, Assumptions (C), (D) or their simplified version (CD) impose decay conditions
on the cumulants of different entries of \(W\), we refer the reader to~\cite[Eqs. (2.5a)--(2.5b)]{MR3941370} for the precise condition. The self-consistent density of states \(\rho\) is defined analogously to~\eqref{scdos}, where \(M\) solves the MDE~\eqref{MDE} with \(\braket{M}\) replaced by \(\mathcal{S}[M]\).} defined in~\cite{MR3941370}. In the second situation,~\ref{case2} in Theorem~\ref{thm2b} below we consider deformed Wigner matrices \(H\) and general \(A\) with a condition complementary to the condition \(\braket{\mathring A^2}\ge c_0\) of Theorem~\ref{thm1}.

\begin{theorem}[Quenched monoparametric universality via spectral sampling mechanism]\label{thm2b}
   There is a small universal constant \(a_0\) and for any small \(c_1>0\) there exists a \(c_0>0\) such that the following hold. Let \(x= N^{-a}\chi\) be a scalar random variable satisfying Assumption~\ref{ass:x}, and let \(H,A,a\) be such that either
   \begin{enumerate}[label=Case \arabic*)]
      \item\label{case1} \(H\) is a correlated random matrix\textsuperscript{\ref{cor mat def}}, \(A=I\), and \(a\in [0, 1-a_1]\) for an arbitrary small \(a_1\),
      \item\label{case2} \(H\) is a deformed Wigner matrix (cf.\ Ass.~\ref{ass:entr}), \(c_0\abs{\braket{A}}\ge  \braket{\mathring{A}^2}^{1/2} \), \(\abs{\braket{A}}\ge c_0\), and \(a\in [0,a_0]\),
   \end{enumerate}
   and fix an energy \(E\) with \(\rho^x(E)\ge c_1>0\) for \(\Prob_x\)-almost all \(x\). Then there exist small \(\alpha,\kappa>0\) and an event \(\Omega_E=\Omega_{E,A}\) with \(\Prob_H(\Omega_E^c)\le N^{-\kappa}\) such that for all \(H\in \Omega_E\) the matrix \(H^x\) satisfies
   \begin{equation}\label{eq:mainbneedsam}
      \left|\E_x f\left(N\rho^x(E)[\lambda_{i_0(x,E)+1}^x-\lambda_{i_0(x,E)}^x]\right)-\int p_\beta(t)f(t)\dif t\right|=\mathcal{O}\left(N^{-\alpha} \| f\|_{C^5}\right).
   \end{equation}
   The exponents \(\kappa, \alpha\) depend on \(a_0, a_1\) while the
   implicit constant in~\eqref{eq:mainbneedsam} depends on \(a_0,a_1, c_0,c_1\) and the diameter of $\supp f$.
\end{theorem}
We remark that the condition \(c_0|\braket{A}|\ge  \braket{\mathring{A}^2}^{1/2} \) in~\ref{case2} is not really necessary for~\eqref{eq:mainbneedsam} to hold. Indeed, if \(\braket{\mathring{A}^2}\ge c\) with any small positive constant, then we are back to the setup of Theorem~\ref{thm1} where the \emph{eigenbasis rotation} mechanism is effective. One can easily see that the proof of~\eqref{univ1} implies~\eqref{eq:mainbneedsam} in this case (see Remark~\ref{rem:betexp} below). However, we kept the condition \(c_0|\braket{A}|\ge  \braket{\mathring{A}^2}^{1/2} \) with a sufficiently small \(c_0\) in the formulation of  Theorem~\ref{thm2b} since it is necessary for the \emph{spectral sampling} mechanism to be effective
which is the mechanism represented in Theorem~\ref{thm2b}.

Note that as long as \(\mathring{A}\ne 0\), i.e.~\ref{case1} is not applicable,
the eigenbasis of \(H^x\) changes with \(x\) and we have to rely on the multi-resolvent local law method.
However,  lacking an effective lower bound on \(\braket{\mathring{A}^2}\),
the effective asymptotic orthogonality  still comes from the spectral sampling effect of
\(\braket{A}\),  the nontrivial tracial part of \(A\).
So along the proof of~\ref{case2}, technically we follow the eigenbasis rotation mechanism, but morally the effect
is similar to the spectral sampling mechanism as it still comes
from a shift in the spectrum triggered by \(x\braket{A}\), the leading part of \(xA\) in \(H^x=H+xA\).
Finally, a
simple perturbation argument shows that \(x\mathring{A}\) has no sizeable effect on the
sampling, but its presence  hinders the technically simpler orthogonality proof used in~\ref{case1}.

\subsection{Gaudin-Mehta distribution}\label{sec painleve}
For completeness we close this section by providing explicit formulas for the universal \emph{Gaudin-Mehta gap distributions} \(p_1,p_2\) which can either be defined as the Fredholm determinant of the sine kernel~\cite{MR0220494} or via the solution to the Painlev\'e V differential equation~\cite{MR573370}.
Given the solution \(\sigma\) to the non-linear differential equation
\begin{equation}
   (t\sigma'')^2 + 4(t\sigma'-\sigma)(t\sigma'-\sigma+(\sigma')^2)=0, \qquad \sigma(t)\sim-\frac{t}{\pi}-\frac{t^2}{\pi^2}\quad(\text{as \(t\to0\)}),
\end{equation}
we have~\cite{MR1808248}
\begin{equation}\label{pbeta eq}
   p_2(s)= \frac{\dif^2}{\dif s^2} \exp\Bigl(\int_0^{\pi s}\frac{\sigma(t)}{t}\dif t\Bigr),\quad p_1(s) = \frac{\dif^2}{\dif s^2} \exp\Bigl(\frac{1}{2}\int_0^{\pi s}\Bigl( \frac{\sigma(t)}{t} - \sqrt{-\frac{\dif}{\dif t} \frac{\sigma(t)}{t} } \Bigr)\dif t\Bigr).
\end{equation}
Remarkably, the Wigner surmise
\begin{equation}
   p_2^{\mathrm{Wigner}}(s):= \frac{32s^2}{\pi^2} \exp\Bigl(-\frac{4s^2}{\pi}\Bigr),\qquad p_1^{\mathrm{Wigner}}(s):= \frac{s\pi}{2} \exp\Bigl(-\frac{\pi s^2}{4}\Bigr)
\end{equation}
obtained by E.\ Wigner from explicitly computing the gap distribution for \(2\times2\) matrices, is very close to the large \(N\) limit \(p_2(s)\), more precisely \(\sup_s \abs{p_2(s)-p_2^{\mathrm{Wigner}}(s)}\approx 0.005\) and \(\sup_s \abs{p_1(s)-p_1^{\mathrm{Wigner}}(s)}\approx 0.016\).

\section{Quenched universality: Proof of Theorem~\ref{thm1} and Theorem~\ref{thm2b}}
In Section~\ref{sec:univrot} we prove Theorem~\ref{thm1} while in Section~\ref{sec:univsamp}
we present the proof of Theorem~\ref{thm2b} which structurally is analogous to the argument in Section~\ref{sec:univrot}.
For notational simplicity we introduce the discrete difference operator \(\delta\), i.e.\ for a tuple \(\lambda\) we set
\begin{equation}
   (\delta\lambda)_i=\delta\lambda_i :=  \lambda_{i+1}-\lambda_i
\end{equation}
in order to express eigenvalue differences \emph{(gaps)} more compactly. We also introduce the notation \(\braket{f}_\mathrm{gap}\) for the expectation of test functions \(f\) with respect to the density \(p_\beta\) from~\eqref{pbeta eq}, i.e.
\begin{equation}
   \braket{f}_\mathrm{gap}:=\int p_\beta(t)f(t)\dif t.
\end{equation}

\subsection{Universality via eigenbasis rotation mechanism: Proof of Theorem~\ref{thm1}}\label{sec:univrot}

To prove Theorem~\ref{thm1} we will show that the gaps \(\lambda_{i+1}^{x_1}-\lambda_i^{x_1}\), \(\lambda_{i+1}^{x_2}-\lambda_i^{x_2}\) for sufficiently large \(|x_1-x_2|\) are asymptotically independent in the sense of the following proposition whose proof will be presented in Section~\ref{sec:DBM}.
In the following we will often denote the covariance of two random variables \(X, Y\) in the \(H\)-space by
\[
   \Cov_H(X,Y):=\E_H XY-(\E_H X)(\E_H Y).
\]
\begin{proposition}\label{pro:realmaintres} Under the conditions of Theorem~\ref{thm1}
   there exists a sufficiently small  \(c^*>0\)  (depending on \(c_0, c_1\)) and
   for any small \(\zeta_1\) there exists \(\zeta_2>0\) such that the following holds. Pick real numbers
   \(x_1,x_2\) with \(N^{-\zeta_1}\le |x_1-x_2|\le c^*\)
   and indices \(j_1,j_2\) with \(|j_1-j_2|\lesssim  N|x_1-x_2|\), such that the corresponding quantiles \(\gamma_{j_r}^{x_r}\), are in the \(c_1\)-bulk of the spectrum of \(H^{x_r}\) for each \(r=1,2\). Then the covariance \(\Cov_H(X,Y)\) satisfies
   \begin{equation}\label{eq:neednow}
      \Cov_H\Bigl(P_1(N\delta\lambda_{j_1}^{x_1}),P_2(N\delta\lambda_{j_2}^{x_2})\Bigr)=\landauO*{N^{-\zeta_2} \norm{P_1}_{C^5}\norm{P_2}_{C^5}}
   \end{equation}
   for any \(P_1,P_2\colon\R\to\R\) bounded smooth test functions, and where %
   the implicit constant in \(\mathcal{O}(\cdot)\) may depend on \(c_0,c_1\) at most polynomially.
\end{proposition}

\begin{remark}\label{rem:newrem}
   We stated the asymptotic independence of a single gap in Proposition~\ref{pro:realmaintres} and only for two \(x_1,x_2\) for notation simplicity. Exactly the same proof as in Section~\ref{sec:DBM} directly gives the result in~\eqref{eq:neednow}  for test functions \(P_r:\R^p\to \R\) of several gaps, for some fixed \(p\in\N\).
   Additionally, by the same proof we can also conclude the asymptotic independence of several gaps for several \(x_1,\dots, x_q\). For the same reason we also state Proposition~\ref{pro:realmaintressam} and Proposition~\ref{pro:realmaintresa} below only for two \(x_1,x_2\) and test functions \(P_r:\R\to \R\).
\end{remark}

\begin{proof}[Proof of Theorem~\ref{thm1}]
   We will first prove that without loss of generality we may assume that the
   linear size of the support of \(x\) is bounded by \(c^*\),
   where \(c^*\) is from Proposition~\ref{pro:realmaintres}.
   This initial simplification will then allow us to use perturbation in \(x\) when proving Proposition~\ref{pro:realmaintres}.
   Suppose that Theorem~\ref{thm1} is already proved for random variables with such a small support
   with an error term \(N^{-\alpha}\) on sets of probability at least \(1-N^{-2\kappa}\)
   and we are now given a random variable \(x\) with a larger support of size bounded by some  constant
   \(C\).  Then
   we define the random variables
   \[
      x_i:=\frac{x\cdot\bm1(x\in J_i)}{\Prob_x(J_i)},
   \]
   where \(J_i\)'s, for \(i=1,2, \ldots , C/c^*\), are disjoint intervals of size \(c^*\) such that \(\supp(x)=\bigcup_i J_i\). For any test function \(f\) we can then write
   \begin{equation}\label{eq:redsmall}
      \begin{split}
         \E_x f(N\rho^x(\gamma_j^x)\delta\lambda_j^x)&=\sum_{i=1}^{C/c^*}\Prob_x(J_i) \E_{x_i}f(N\rho^{x_i}(\gamma_j^{x_i})\delta\lambda_j^{x_i}) \\
         &=\braket{f}_\mathrm{gap}\sum_{i=1}^{C/c^*}\Prob_x(J_i)+\mathcal{O}(N^{-\alpha})=\braket{f}_\mathrm{gap}+\landauO*{N^{-\alpha}},
      \end{split}
   \end{equation}
   on a set of probability at least \(1-(C/c^*) N^{-2\kappa}\ge 1-N^{-\kappa}\),
   where we used Theorem~\ref{thm1} for the random variables \(x_i\) in the last step and a union bound.

   From now on we assume that the linear size of the support of \(x\) is bounded by  \(c^*\).
   With \(\nu(\dif x)\) denoting the measure of \(x\) we have
   \begin{align}\label{eq:mainvarcomp}
       & \E_H\left|\E_x f\left(N\rho^x(\gamma_{j}^x)\delta\lambda_j^x\right)-\braket{f}_\mathrm{gap}\right|^2 \\\nonumber
       & =\iint_{\abs{x_1-x_2}\ge N^{-\epsilon_2}}\nu(\dif x_1) \nu(\dif x_2)\E_H
      \Big[ \prod_{r=1}^2f\left(N\rho^{x_r}(\gamma_{j}^{x_r})\delta\lambda_{j}^{x_r}\right)
         \Big] -\braket{f}_\mathrm{gap}^2+\landauO*{N^{-c(\epsilon_2)}\norm{f}_{C^5}^2 },
   \end{align}
   for some sufficiently small \(\epsilon_2\) so that we can apply Proposition~\ref{pro:realmaintres}
   with \(\zeta_1 = \epsilon_2\).
   In~\eqref{eq:mainvarcomp} we used that
   the regime \(|x_1-x_2|\le N^{-\epsilon_2}\) can be removed at the price of a negligible error by the regularity assumption on the distribution of \(x=N^{-a}\chi\), with \(\chi\) satisfying Assumption~\ref{ass:x}. For the cross-term in~\eqref{eq:mainvarcomp} we used that by gap universality for the deformed Wigner matrix \(H^x\) with a fixed \(x\) (see e.g.~\cite[Corollary 2.11]{MR3916109}) it follows that
   \begin{equation}\label{eq H gap univ}
      \E_H f\left(N\rho^x(\gamma_{j}^x)\delta\lambda_{j}^x\right)=\braket{f}_\mathrm{gap}+\landauO*{N^{-\zeta_3}\norm{f}_{C^5}}
   \end{equation}
   for some small fixed \(\zeta_3>0\) depending only on the model parameters and on the constants \(a_0,c_1\).

   Applying Proposition~\ref{pro:realmaintres}
   to the first term in~\eqref{eq:mainvarcomp} with \(P_r(t):=f(\rho^{x_r}(\gamma_j^{x_r})t)\) noting that \(\rho^{x_r}\) is uniformly bounded, so that for \( N^{-\epsilon_2}\le \abs{x_1-x_2}\le c^*\), we get
   \begin{equation}\label{esq cov est}
      \E_H \Bigl[ \prod_{r=1}^2f\left(N\rho^{x_r}(\gamma_{j}^{x_r})\delta\lambda_{j}^{x_r}\right) \Bigr] = \prod_{r=1}^2\E_H f\left(N\rho^{x_r}(\gamma_{j}^{x_r})\delta\lambda_{j}^{x_r}\right) + \landauO*{N^{-\zeta_2}\norm{f}_{C^5}^2}.
   \end{equation}
   By using~\eqref{esq cov est} and~\eqref{eq H gap univ} in~\eqref{eq:mainvarcomp} it follows that
   \begin{equation}\label{eq:fincheba}
      \E_H\left|\E_x f\left(N\rho^x(\gamma_{j}^x)\delta\lambda_{j}^x\right)-\braket{f}_\mathrm{gap}\right|^2\le \left(N^{-c(\epsilon_2)}+N^{-\zeta_2}+N^{-\zeta_3}\right)\norm{f}_{C^5}^2.
   \end{equation}
   From~\eqref{eq:fincheba} and the Chebyshev inequality we obtain events \(\Omega_{j,f}\) on which~\eqref{univ1} holds with probability \(\Prob_H(\Omega_{j,f}^c)\le N^{-\kappa}\) for some suitably chosen \(\kappa,\alpha>0\).
\end{proof}

\subsection{Universality via spectral sampling mechanism: Proof of Theorem~\ref{thm2b}}\label{sec:univsamp}

The mechanism behind the proof of Theorem~\ref{thm2b} is quite different compared to Theorem~\ref{thm1}. In particular, in order to prove Theorem~\ref{thm2b} we will first show that under the assumptions of Theorem~\ref{thm1} the gaps \(\delta\lambda_i^x\), \(\delta\lambda_j^x\) are asymptotically independent for any fixed \(x\) in the probability space of \(H\) as long as \(\abs{i-j}\) is sufficiently large. This independence property
for the  \(A=I\) case has already been
used  as a heuristics without proof, e.g.\ in~\cite{MR3112927,MR4416591}
(a related result for not too distant  gaps for local log-gases
can be deduced from the De Giorgi-Nash-Moser H\"older regularity estimate, see~\cite[[Section 8.1]{MR3372074}).
More precisely, we have the following proposition:
\begin{proposition}\label{pro:realmaintressam}
   Under the conditions of Theorem~\ref{thm2b}
   there exists a sufficiently small  \(c_*>0\)  (depending on \(c_0, c_1\)) and
   for any sufficiently small \(\zeta_1>0\) there exists \(\zeta_2>0\)
   such that %
   the following hold.
   Pick indices \(j_1, j_2\)  and real numbers
   \(x_1,x_2\)  such that the corresponding quantiles \(\gamma_{j_r}^{x_r}\) are in  the \(c_1\)-bulk of the spectrum of \(H^{x_r}\), i.e.
   \(\rho^{x_r}(\gamma_{j_r}^{x_r})\ge c_1\),
   for \(r=1,2\). In the two different cases listed in Theorem~\ref{thm2b}
   we additionally assume the following:
   \begin{itemize}
      \item[\ref{case1}] \(|j_1-j_2|\ge  N^{\zeta_1} \);
      \item[\ref{case2}] \(N^{1-\zeta_1}\le |j_1-j_2|\le c_*N\) and  \(N|x_1-x_2|\lesssim |j_1-j_2|\).
   \end{itemize}
   Then in both cases it holds that
   \begin{equation}
      \Cov_H\Bigl( P_1(N \delta\lambda_{j_1}^{x_1}),P_2(N \delta\lambda_{j_2}^{x_2})\Bigr) = \mathcal{O}\left(N^{-\zeta_2}\prod_{r=1}^2 \norm{P_r}_{C^5}\right)
   \end{equation}
   for \(P_r\colon\R\to\R\) bounded, smooth test functions. The implicit constant in \(\mathcal{O}(\cdot)\) may depend on \(c_0,c_1\) at most polynomially.
\end{proposition}

\begin{proof}[Proof of Theorem~\ref{thm2b}]
   We present the proof only for the more involved~\ref{case2}, the~\ref{case1} is much easier and omitted.
   Similarly to~\eqref{eq:fincheba} in the proof of Theorem~\ref{thm1} it is enough to consider
   the case when the linear size of the support of \(x\) is bounded by some \(\tilde c>0\) (determined later) and prove that
   \begin{equation}\label{eq:afterchebsam}
      \E_H\left|\E_x f\left(N\rho^x(E)\delta\lambda_{i_0(x,E)}\right)-\braket{f}_\mathrm{gap}\right|^2\le N^{-2\alpha-\kappa}\| f\|_{C^5}^2,
   \end{equation}
   for some small \(\alpha,\kappa>0\).

   Note that under the assumptions of~\ref{case2} in
   Theorem~\ref{thmA} for any \(x_1, x_2\) in the support of the random variable \(x\) it holds
   \begin{equation}\label{eq:farindex}
      \theta |i_0(x_1,E)-i_0(x_2,E)|\le  N |x_1-x_2| \le \Theta |i_0(x_1,E)-i_0(x_2,E)|
   \end{equation}
   with some \(\theta, \Theta\) (depending on \(c_0, c_1\))
   as long as \(|x_1-x_2|\gg N^{-1}\). The bound in~\eqref{eq:farindex} is a direct consequence of the following Lemma~\ref{lem:gamma}
   (assuming that \(\tilde c\le c^*\))
   whose proof is postponed to Appendix~\ref{sec:M}.

   \begin{lemma}\label{lem:gamma}
      For any \(c_1>0\), there exists a \(c^*=c^*(c_1)>0\) such that for
      \(x_1,x_2\) with \(|x_1-x_2|\le c^*\)  it holds that
      \begin{equation}\label{eq:approxgamma}
         \gamma_i^{x_1}=\gamma_i^{x_2}+(x_1-x_2)\braket{A}+\mathcal{O}\left(|x_1-x_2|\braket{\mathring{A}^2}^{1/2}+|x_1-x_2|^2\right),
      \end{equation}
      where \(\gamma_i^{x_r}\) are the quantiles of \(\rho^{x_r}\) and \(i\) is in the \(c_1\)-bulk, i.e.
      \(\rho^{x_r}(\gamma_i^{x_r})\ge c_1\), for \(r\in [2]\).
   \end{lemma}
   Indeed,~\eqref{eq:farindex} follows from \(|\braket{A}|\ge c_0\) and from the inequality
   \[
      \begin{split}
         \left|\frac{i_0(x_1,E)}{N}-\frac{i_0(x_2,E)}{N}\right|&=\left|\int_{\gamma_{i_0(x_2,E)}^{x_1}}^{\gamma_{i_0(x_1,E)}^{x_1}}\rho^{x_1}(E)\,\dif E\right|\ge c_1\left|\gamma_{i_0(x_2,E)}^{x_1}-\gamma_{i_0(x_1,E)}^{x_1}\right| \\
         &\ge \frac{c_1}{2} \left|\gamma_{i_0(x_2,E)}^{x_1}-\gamma_{i_0(x_2,E)}^{x_2}\right|  + \mathcal{O}(N^{-1})
         \ge \frac{c_1}{4}|x_1-x_2| \, |\braket{A}|,
      \end{split}
   \]
   where to go from the first to the second line we used that \(\gamma_{i_0(x_2,E)}^{x_2}=\gamma_{i_0(x_1,E)}^{x_1}+\mathcal{O}(N^{-1})\) by the definition of \(i_0(x,E)\)  and that we are in the bulk.
   In the last inequality  we used~\eqref{eq:approxgamma}
   and that  its error terms  are negligible  by \(c_0|\braket{A}|\ge \braket{\mathring{A}^2}^{1/2}\)
   and \(|x_1-x_2|\le \tilde c\)  assuming that \(\tilde c\le c_0/10\).
   Then by a similar chain of inequalities, and using~\eqref{eq:approxgamma} once more, we get the matching upper bound in~\eqref{eq:farindex}.

   To prove~\eqref{eq:afterchebsam}, we use the counterpart of~\eqref{eq:mainvarcomp}
   and  that we can neglect the regime \(|x_1-x_2|\le N^{-\epsilon_2}\)  for some  sufficiently  small
   \(\epsilon_2>0\) so that we can apply Proposition~\ref{pro:realmaintressam} with
   \(\zeta_1=\epsilon_2\). %
   We remove this regime to ensure that on its complement
   \(|i_0(x_1,E)-i_0(x_2,E)|\) is sufficiently large by~\eqref{eq:farindex}. More precisely,
   for any \(x_1, x_2\) with \(N^{-\epsilon_2}\le |x_1-x_2|\le \tilde c\),
   we have \(\Theta^{-1}  N^{1-\epsilon_2}\le |i_0(x_1,E)-i_0(x_2,E)|\le \theta^{-1}  \tilde c N\) from~\eqref{eq:farindex}.
   Assuming \(\tilde c\le c_* \theta\),  we can apply~\ref{case2} of  Proposition~\ref{pro:realmaintressam}
   by choosing \(j_1= i_0(x_1,E)\), \(j_2= i_0(x_2,E)\) and
   with exponent \(\zeta_2\) to factorise the expectation in the equivalent of~\eqref{eq:mainvarcomp}.
   Using again the gap universality~\eqref{eq H gap univ}, similarly to~\eqref{eq:fincheba} we conclude~\eqref{eq:mainbneedsam} choosing \(\alpha, \kappa>0 \) appropriately.
\end{proof}

\begin{remark}\label{rem:betexp}
   The proof of~\eqref{eq:mainbneedsam} in the case \(\braket{\mathring{A}^2}\ge c\) is analogous to the proof of Theorem~\ref{thm1} above. We note that Proposition~\ref{pro:realmaintres} allows to also conclude the asymptotic independence of \(\delta\lambda_{i_0(x_1,E)}\) and \(\delta\lambda_{i_0(x_2,E)}\) since \(|i_0(x_1,E)-i_0(x_2,E)|\lesssim N |x_1-x_2|\) due to \eqref{eq:farindex}.
\end{remark}
\subsection{Proof of Corollary~\ref{thm2}}\label{appendix prop}
Picking \(E=0\) and the test function\footnote{While \(f\) does not literally satisfy the regularity condition, one can easily extend the validity of~\eqref{eq:mainbneedsam} to interval characteristic functions \(f\) by a standard approximation argument.} \(f(u) =\bm1(0\le u\le y)\) in~\ref{case1} of Theorem~\ref{thm2b}, and choosing the random variable \(x\) such that
\(-x\) has density proportional to \(\rho\vert_I\), with \(\rho=\rho_{\mathrm{sc}}\), it follows that with very high probability in the space of \(H\) it holds that
\begin{equation}
   \begin{split}
      &\E_x F(N\rho^x(0)\delta\lambda^x_{i_0(x,0)}) \\
      &\quad=\E_x F(N\rho(\lambda_{i_0(0,-x)})\delta\lambda_{i_0(0,-x)})+\landauO{N^{-1+\xi}}\\
      &\quad =\Bigl(\int_I\rho\Bigr)^{-1}\sum_i F(N\rho(\lambda_i) \delta\lambda_{i}) \int_{I\cap (\gamma_{i-1},\gamma_i]} \rho(t) \dif t+\landauO{N^{-1+\xi}}\\
      &\quad=\Bigl(N\int_I\rho\Bigr)^{-1} \#\set*{i\given  \lambda_{i+1}-\lambda_i\le \frac{y}{N\rho(\lambda_i)}, \lambda_i\in I}  + \landauO*{N^{-1+\xi}\abs{I}^{-1}},
   \end{split}
\end{equation}
where in the first and third step we used rigidity (see e.g.~\cite[Lemma 7.1, Theorem 7.6]{MR3068390} or~\cite[Section 5]{MR2871147}), i.e.\ that for any small \(\xi>0\) we have
\begin{equation}\label{eq:rig}
   \abs{\lambda_i-\gamma_i}\le \frac{N^\xi}{N},
\end{equation}
for all \(\gamma_i\) in the bulk, with very high probability.

\section{DBM analysis: Proof of Proposition~\ref{pro:realmaintres} and~Proposition~\ref{pro:realmaintressam}}\label{sec:DBM}
In this section we first present the proof of Proposition~\ref{pro:realmaintres} in details and
later in Section~\ref{sec:36} we
explain  the very minor changes that are required to prove Proposition~\ref{pro:realmaintressam}. %

\subsection{Proof of Proposition~\ref{pro:realmaintres}}\label{sec:mainsec}
By standard Green function comparison (GFT) argument (see e.g.~\cite[Section 15]{MR3699468}) it is enough to prove Proposition~\ref{pro:realmaintres} only for matrices with a small Gaussian component. More precisely, consider the DBM flow
\begin{equation}\label{eq:matrdbm}
   \dif \widehat{H}_t=\frac{\dif \widehat{B}_t}{\sqrt{N}}, \qquad \widehat{H}_0=H^\#,
\end{equation}
with \(\widehat{B}_t\) being a real symmetric or complex Hermitian standard Brownian motion (see e.g.~\cite[Section 12.1]{MR3699468} for the precise definition) independent of the initial condition \(\widehat{H}_0=H^\#\),
where \(H^\#\) is a deformed Wigner matrix specified later. Throughout this section we fix \(T>0\) and analyse the DBM for times \(0\le t\le T\).

We denote the ordered collection of eigenvalues of \(\widehat{H}_t+xA\) by \({\bm \lambda}^x(t)=\{\lambda_i^x(t)\}_{i\in [N]}\). The main result of this section is the asymptotic independence of \({\bm \lambda}^{x_1}(t_1)\), \({\bm \lambda}^{x_2}(t_1)\) for \(|x_1-x_2|\ge N^{-\zeta_1}\) and \(t_1\ge N^{-1+\omega_1}\), for some \(\omega_1>0\).

We note that in this entire section we do not use the randomness of \(x\), in the statement of Propositions~\ref{pro:realmaintres} and~\ref{pro:realmaintressam} \(x_1, x_2\)  are fixed parameters. Hence all probabilistic statements, such as covariances etc., are understood in the probability space of
the random matrices and the driving Brownian motions in~\eqref{eq:matrdbm}.
\begin{proposition}\label{pro:realmaintresa}
   Let \(H^\#\) be a deformed Wigner matrix satisfying Assumption~\ref{ass:entr}, let \(\widehat{H}_t\) be the solution of~\eqref{eq:matrdbm}, and let \(A\)  be  a deterministic matrix such that \(\braket{\mathring{A}^2}\ge c_0\)  and \(\| A\|\lesssim 1\). Then there exists a small  \(c^*>0\)  (depending on \(c_0, c_1\)) and
   for any small \(\zeta_1,\omega_1>0\) there exists some \(\zeta_2>0\) such that
   the following hold. Fix \(x_1,x_2\) with \(N^{-\zeta_1}\le |x_1-x_2|\le c^*\)  and
   indices \(j_1, j_2\) such that \(|j_1-j_2|\lesssim N|x_1-x_2|\),
   and the corresponding quantiles \(\gamma_{j_r}^{x_r}\) are in the bulk of the spectrum of \(H^\#+x_r A\) for \(r=1,2\). Then for the eigenvalues of \(\widehat{H}_t+x_r A\) it holds that
   \begin{equation}\label{eq:factora}
      \Cov \left[P\left(N \delta\lambda_{j_1}^{x_1}(t_1)\right), Q\left(N\delta\lambda_{j_2}^{x_2}(t_1)\right)\right]=\mathcal{O}\left(N^{-\zeta_2} \norm{P}_{C^1}\norm{Q}_{C^1} \right),
   \end{equation}
   with \(t_1=N^{-1+\omega_1}\) for any \(P,Q\colon\R\to\R\) bounded smooth test functions.

\end{proposition}

Using Proposition~\ref{pro:realmaintresa} as an input we readily conclude Proposition~\ref{pro:realmaintres}.
\begin{proof}[Proof of Proposition~\ref{pro:realmaintres}]
   Let \(H\) be the deformed Wigner matrix from Proposition~\ref{pro:realmaintres}, and consider the Ornstein-Uhlenbeck flow
   \begin{equation}\label{eq:matrou}
      \dif H_t=-\frac{1}{2}(H_t-\E H_0) \dif t+\frac{\dif B_t}{\sqrt{N}}, \qquad H_0=H,
   \end{equation}
   with \(B_t\) being a real symmetric or complex Hermitian standard Brownian motion independent of \(H_0\).

   Let \(H^\#_{t_1}\), with \(t_1\) from Proposition~\ref{pro:realmaintresa}, be such that
   \begin{equation}\label{eq:solou}
      H_{t_1}\stackrel{\mathrm{d}}{=}H^\#_{t_1}+\sqrt{c(t_1)t_1}U,
   \end{equation}
   with \(U\) a GOE/GUE matrix independent of \(H^\#_{t_1}\) and \(c=c(t_1)=1+\mathcal{O}(t_1)\) is an appropriate constant very close to one. Then by~\eqref{eq:solou} it follows that
   \begin{equation}\label{eq:imprel}
      \widehat{H}_{ct_1}\stackrel{\mathrm{d}}{=} H_{t_1},
   \end{equation}
   with \(\widehat{H}_{ct_1}\) being the solution of~\eqref{eq:matrdbm} with initial condition \(\widehat{H}_0=H^\#_{t_1}\).

   Then, by a standard GFT argument~\cite[Section 15]{MR3699468}, we have that
   \begin{equation}\label{eq:impGFT}
      \begin{split}
         &\Cov\left(P\left(N\delta\lambda_{j_1}^{x_1}\right), Q\left(N\delta\lambda_{j_2}^{x_2}\right)\right)\\
         &\qquad=\Cov\left(P\left(N\delta\lambda_{j_1}^{x_1}(ct_1)\right), Q\left(N\delta\lambda_{j_2}^{x_2}(ct_1)\right)\right)+
         \mathcal{O}\left(N^{-\zeta_2}\norm{P}_{C^5}\norm{Q}_{C^5}\right).
      \end{split}
   \end{equation}
   Finally, by~\eqref{eq:impGFT} together with~\eqref{eq:imprel} and Proposition~\ref{pro:realmaintresa} applied to \(H^\#:=  H^\#_{t_1}\) we conclude the proof of Proposition~\ref{pro:realmaintres}.
\end{proof}

\subsection{Proof of Proposition~\ref{pro:realmaintresa}}\label{sec:DBMDBM}
This proof is an adaptation of the proof of~\cite[Proposition 7.2]{1912.04100}
(which itself is based upon~\cite{MR3914908}) with two minor differences. First, the DBM in this paper is for eigenvalues (see~\eqref{eq:eigflow} below) while in~\cite[Eq. (7.15)]{1912.04100} it was for singular values.
Second,  in~\cite[Section 7]{1912.04100} it was sufficient to consider singular values close to zero hence
the base points \(j_1\) and \(j_2\) were fixed to be \(0\); here they are arbitrary.  Both changes
are simple to incorporate, so we present only the backbone of the
proof that shows the differences,  skipping certain steps that remain unaffected.

The flow~\eqref{eq:matrdbm} induces the following flow on the eigenvalues of \(\widehat{H}_t+x_l A\):
\begin{equation}\label{eq:eigflow}
   \dif\lambda_i^{x_r}(t)=\sqrt{\frac{2}{\beta N}}\dif b_i^{x_r}(t)+\frac{1}{N}\sum_{j\ne i}\frac{1}{\lambda_i^{x_r}(t)-\lambda_j^{x_r}(t)}\dif t,
\end{equation}
with \(r\in [2]\) and \(\beta=1, 2\) in the real and complex case, respectively. Here (omitting the time dependence) we used the notation
\begin{equation}
   \dif b_i^{x_r}=\sqrt{\frac{2}{\beta}}\sum_{a,b=1}^N \overline{{\bm u}_i^{x_r}(a)}\dif \widehat{B}_{ab}(t){\bm u}_i^{x_r}(b),
\end{equation}
with \({\bm u}_i^{x_r}(t)\) being the orthonormal eigenvectors of \(\widehat{H}_t+x_r A\). The collection \({\bm b}^{x_r}:=\{b_i^{x_r}\}_{i\in [N]}\), for fixed \(r\), consists of i.i.d standard real Brownian motions. However, the families \({\bm b}^{x_1}\), \({\bm b}^{x_2}\) are not independent for different \(r\)'s, in fact their joint distribution is not necessarily Gaussian. The quadratic covariation of these two processes is given by
\begin{equation}\label{eq:quadcov12}
   \dif[b_i^{x_1}(t),b_j^{x_2}(t)]=\big|\braket{{\bm u}_i^{x_1}(t),{\bm u}_j^{x_2}(t)}\big|^2\dif t.
\end{equation}
We remark that in~\eqref{eq:quadcov12} we used a different notation for the quadratic covariation compared to~\cite[Section 7.2.1]{1912.04100}.

\subsubsection{Definition of the comparison processes for \({\bm \lambda}^{x_r}\)}

To make the notation cleaner we only consider the real case (\(\beta=1\)). To prove
the asymptotic independence of
the processes \({\bm \lambda}^{x_1}\), \({\bm \lambda}^{x_2}\), realized on the probability space \(\Omega_b\),
we will compare them with two completely independent processes \({\bm \mu}^{(r)}(t)=\{\mu_i^{(r)}(t)\}_{i=1}^N\) realized on a different probability space \(\Omega_\beta\). The processes \({\bm \mu}^{(r)}(t)\) are the unique strong solution of
\begin{equation}\label{eq:indeppro}
   \dif \mu_i^{(r)}(t)=\sqrt{\frac{2}{N}}\dif \beta_i^{(r)}+\frac{1}{N}\sum_{j\ne i}\frac{1}{\mu_i^{(r)}(t)-\mu_j^{(r)}(t)}\dif t, \qquad \mu_i^{(r)}(0)=\mu_i^{(r)},
\end{equation}
with \(\mu_i^{(r)}\) being the eigenvalues of two independent GOE matrices \(H^{(r)}\), and \({\bm \beta}^{(r)}=\{\beta_i^{(r)}(t)\}_{i=1}^N\) being independent vectors of standard i.i.d. Brownian motions.

We now define two intermediate processes \(\widetilde{\bm \lambda}^{(r)}(t)\), \(\widetilde{\bm \mu}^{(r)}(t)\) so that for \(t\gg N^{-1}\) the particles \(\widetilde{\lambda}_i^{(r)}(t)\), \(\widetilde{\mu}_i^{(r)}(t)\) will be close to \(\lambda_i^{x_r}(t)\) and \(\mu_i^{x_r}(t)\), respectively, for indices \(i\) close to \(j_r\), with very high probability (see Lemmas~\ref{lem:lam}--\ref{lem:mu} below). Additionally, the processes \(\widetilde{\bm \lambda}^{(r)}(t)\), \(\widetilde{\bm \mu}^{(r)}(t)\),  which will be realized on two different probability spaces, will have the same joint distribution:
\begin{equation}\label{eq:samejointdistr}
   \left(\widetilde{\bm \lambda}^{(1)}(t),\widetilde{\bm \lambda}^{(2)}(t)\right)_{0\le t\le T}\stackrel{\mathrm{d}}{=}\left(\widetilde{\bm \mu}^{(1)}(t),\widetilde{\bm \mu}^{(2)}(t)\right)_{0\le t\le T}.
\end{equation}

Fix any small \(\omega_A>0\) (later \(\omega_A\) will be chosen smaller than \(\omega_E\) from~\eqref{eq:bover}) and define the process \(\widetilde{\bm \lambda}^{(r)}(t)\) to be the unique strong solution of
\begin{equation}\label{eq:complam}
   \dif\widetilde{\lambda}_i^{(r)}(t)=\frac{1}{N}\sum_{j\ne i}\frac{1}{\widetilde{\lambda}_i^{(r)}(t)-\widetilde{\lambda}_j^{(r)}(t)}\dif t+\begin{cases}
      \sqrt{\frac{2}{N}}\dif b_i^{x_r}             & \text{if}\quad |i-j_r|\le N^{\omega_A}, \\
      \sqrt{\frac{2}{N}}\dif \widetilde{b}_i^{(r)} & \text{if}\quad |i-j_r|> N^{\omega_A},
   \end{cases}
\end{equation}
with initial data \(\widetilde{\bm \lambda}^{(r)}(0)\) being the eigenvalues of independent GOE matrices, which are also independent of \(H^\#\) in~\eqref{eq:matrdbm}.
Here the Brownian motions
\begin{equation}
   \underline{\bm b}^\mathrm{in} := (b_{j_1-N^{\omega_A}}^{x_1},\ldots, b_{j_1+N^{\omega_A}}^{x_1}, b_{j_2-N^{\omega_A}}^{x_2},\dots, b_{j_2+N^{\omega_A}}^{x_2}).
\end{equation}
for indices close to \(j_r\) are exactly the ones in~\eqref{eq:eigflow}. For indices away from \(j_r\) we define the driving Brownian motions to be an independent family
\begin{equation}\label{eq:bmout}
   \underline{\bm b}^\mathrm{out}:=\set*{\widetilde{b}_i^{(r)}\given |i-j_r|> N^{\omega_A}, r\in [2]}.
\end{equation} of standard real i.i.d.\ Brownian motions which are also independent of \(\underline{\bm b}^\mathrm{in}\). The Brownian motions \(\underline{\bm b}^\mathrm{out}\) are defined on the same probability space of \(\underline{\bm b}^\mathrm{in}\), which we will still denote by \(\Omega_b\), with a slight abuse of notation.

For any \(i,j\in [4N^{\omega_A}+2]\) we use the notation
\begin{equation}
   i=(r-1)N^{\omega_A}+\mathfrak{i}, \qquad j=(m-1)N^{\omega_A}+\mathfrak{j}
\end{equation}
with \(r,m\in [2]\) and \(\mathfrak{i},\mathfrak{j}\in [2N^{\omega_A}+1]\). The covariance matrix \(C(t)\) of the increments of \(\underline{\bm b}^\mathrm{in}\), consisting of four blocks of size \(2N^{\omega_A}+1\), is given by
\begin{equation}\label{eq:defcov}
   C_{ij}(t)\dif t:=\dif[\underline{\bm b}^\mathrm{in}_i, \underline{\bm b}^\mathrm{in}_j]=\Theta_{\mathfrak{i}\mathfrak{j}}^{x_r,x_m}(t)\dif t,
\end{equation}
where
\begin{equation}\label{eq:thetaover}
   \Theta_{\mathfrak{i}\mathfrak{j}}^{x_r,x_m}(t):= \big|\braket{{\bm u}_{\mathfrak{i}+j_r-N^{\omega_A}-1}^{x_r}(t),{\bm u}_{\mathfrak{j}+j_m-N^{\omega_A}-1}^{x_m}(t)}\big|^2
\end{equation}
and \(\{{\bm u}_i^{x_r}(t)\}_{i=1}^N\) are the orthonormal eigenvectors of \(\widehat{H}_t+x_r A\). Note that \(\{{\bm u}_i^{x_r}(t)\}_{i=1}^N\) are not well-defined if \(\widehat{H}_t+x_r A\) has multiple eigenvectors, however, without loss of generality, we can assume that almost surely \(\widehat{H}_t+x_l A\) does not have multiple eigenvectors for any \(r\in[2]\) for almost all \(t\in [0,T]\) by~\cite[Proposition 2.3]{MR4009717} together with Fubini's theorem. By Doob's martingale representation theorem~\cite[Theorem 18.12]{MR1876169} there exists a real standard Brownian motion \({\bm \theta}(t)\in\R^{4N^{\omega_A}+2}\) such that
\begin{equation}\label{b theta}
   \dif \underline{\bm b}^{\mathrm{in}}=\sqrt{C}\dif {\bm \theta}.
\end{equation}

Similarly, on the probability space \(\Omega_\beta\) we define the comparison process \(\widetilde{\bm \mu}^{(r)}(t)\) to be the solution of
\begin{equation}\label{eq:compmu}
   \dif\widetilde{\mu}_i^{(r)}(t)=\frac{1}{N}\sum_{j\ne i}\frac{1}{\widetilde{\mu}_i^{(r)}(t)-\widetilde{\mu}_j^{(r)}(t)}\dif t+\begin{cases}
      \sqrt{\frac{2}{N}}\dif \zeta_i^{(r)}             & \text{if}\quad |i-j_r|\le N^{\omega_A}, \\
      \sqrt{\frac{2}{N}}\dif \widetilde{\zeta}_i^{(r)} & \text{if}\quad |i-j_r|> N^{\omega_A},
   \end{cases}
\end{equation}
with initial data \(\widetilde{\bm \mu}^{(r)}(0)\) being the eigenvalues of independent GOE matrices defined on the probability space \(\Omega_\beta\), which are also independent of \(H^{(r)}\). We now construct the driving Brownian motions in~\eqref{eq:compmu} so that~\eqref{eq:samejointdistr} is satisfied. For indices away from \(j_r\) the standard real Brownian motions
\begin{equation}
   \underline{\bm\zeta}^\mathrm{out} := \set*{\widetilde{\zeta}_i^{(r)} \given \abs{i-j_r}>N^{w_A}, r\in[2]}
\end{equation}
are i.i.d.\ and they are independent of \({\bm \beta}^{(1)}\), \({\bm \beta}^{(2)}\) in~\eqref{eq:indeppro}. For indices \(|i-j_r|\le N^{\omega_A}\) the collections
\begin{equation}
   \underline{\bm\zeta}^\mathrm{in} := (\zeta_{j_1-N^{\omega_A}}^{(1)},\dots, \zeta_{j_1+N^{\omega_A}}^{(1)}, \zeta_{j_2-N^{\omega_A}}^{(2)},\dots, \zeta_{j_2+N^{\omega_A}}^{(2)})
\end{equation}
will be constructed from the independent families\footnote{The families \(\underline{\bm \beta}^{\mathrm{in}}\), \(\underline{\bm b}^{\mathrm{in}}\) were denoted by \(\underline{\bm \beta}\) and \(\underline{\bm b}\), respectively, in~\cite[Eqs. (7.22)-(7.23)]{1912.04100}.}
\begin{equation}
   \underline{\bm\beta}^{\mathrm{in}}:=(\beta_{j_1-N^{\omega_A}}^{(1)},\dots, \beta_{j_1+N^{\omega_A}}^{(1)}, \beta_{j_2-N^{\omega_A}}^{(2)},\dots, \beta_{j_2+N^{\omega_A}}^{(2)}),
\end{equation}
as follows.

Since the original process \({\bm \lambda}^{x_r}(t)\) and the comparison processes \({\bm \mu}^{(r)}(t)\) are realized on two different probability spaces, we construct a matrix valued process \(C^\#(t)\) and a vector-valued Brownian motion \(\underline{\bm \beta}^{\mathrm{in}}\) on the probability space \(\Omega_\beta\) such that \((C^\#(t),\underline{\bm \beta}^{\mathrm{in}}(t))\) have the same joint distribution as \((C(t),{\bm \theta}(t))\) with \(C,\bm\theta\) from~\eqref{b theta}. This \(\underline{\bm \beta}^{\mathrm{in}}\) is the driving Brownian motion of the \({\bm \mu}^{(r)}(t)\) process in~\eqref{eq:indeppro}. Define the process
\begin{equation}\label{eq:zetain}
   \underline{{\bm \zeta}}^{\mathrm{in}}(t):=\int_0^t \sqrt{C^\#(s)}\dif \underline{\bm \beta}^{\mathrm{in}}(s)
\end{equation}
on the probability space \(\Omega_\beta\). By construction we see that the processes \(\underline{\bm b}^{\mathrm{in}}\) and \(\underline{\bm \zeta}^{\mathrm{in}}\) have the same distribution, and that the two collections \(\underline{\bm b}^{\mathrm{out}}\) and \(\underline{\bm \zeta}^{\mathrm{out}}\) are independent of \(\underline{\bm b}^{\mathrm{in}}\), \(\underline{\bm \beta}^{\mathrm{in}}\) and among each other. Hence we conclude that
\begin{equation}\label{eq:allind}
   \big(\underline{\bm b}^{\mathrm{in}}(t),\underline{\bm b}^{\mathrm{out}}(t)\big)_{0\le t\le T}\stackrel{\mathrm{d}}{=}\big(\underline{\bm \zeta}^{\mathrm{in}}(t),\underline{\bm \zeta}^{\mathrm{out}}(t)\big)_{0\le t\le T}.
\end{equation}
Finally, by the definitions in~\eqref{eq:complam},~\eqref{eq:compmu} and by~\eqref{eq:allind}, we conclude that the processes \(\widetilde{\bm \lambda}^{(r)}(t)\), \(\widetilde{\bm \mu}^{(r)}(t)\) have the same joint distribution (see~\eqref{eq:samejointdistr}), since their initial conditions and their driving processes~\eqref{eq:allind} agree in distribution.

\subsubsection{Proof of the asymptotic independence of the eigenvalues}
In this section we use that the processes \({\bm \lambda}^{x_r}(t)\), \(\widetilde{\bm \lambda}^{(r)}(t)\) and \(\widetilde{\bm \mu}^{(r)}(t)\),  \({\bm \mu}^{(r)}(t)\) are close pathwise at time \(t_1=N^{-1+\omega_1}\) as stated below in Lemma~\ref{lem:lam} and Lemma~\ref{lem:mu}, respectively, to conclude the proof of Proposition~\ref{pro:realmaintresa}. The proof of these lemmas is completely analogous to the proof in~\cite[Lemmas 7.6--7.7]{1912.04100},~\cite[Eq. (3.7), Theorem 3.1]{MR3914908}, hence we will only explain the very minor differences required in this paper. First, we compare the processes \({\bm \lambda}^{x_r}(t)\), \(\widetilde{\bm \lambda}^{(r)}(t)\), in particular this lemma shows that for \(i\) far away from \(j_1,j_2\) the Brownian motions \(b_i^{x_1},b_i^{x_2}\) can be replaced by the independent Brownian motions from \(\underline{\bm b}^{\mathrm{out}}\) at a negligible error.
\begin{lemma}\label{lem:lam}
   Let \({\bm \lambda}^{x_r}(t)\), \(\widetilde{\bm \lambda}^{(r)}(t)\), with \(r\in [2]\), be the processes defined in~\eqref{eq:eigflow} and~\eqref{eq:complam}, respectively. For any small \(\omega_1>0\) there exists \(\omega>0\), with \(\omega\ll \omega_1\), such that it holds
   \begin{equation}
      \left|\rho^{x_r}(\gamma_{j_r}^{x_r})\delta\lambda_{j_r}^{x_r}(t_1)-\rho_{\mathrm{sc}}(\gamma_{j_r})\delta\widetilde{\lambda}_{j_r}^{(r)}(t_1)\right|\le N^{-1-\omega},
   \end{equation}
   for any \(j_r\) in the \(c_1\)-bulk, with very high probability on the probability space \(\Omega_b\), where \(t_1:=N^{-1+\omega_1}\). Here by \(\gamma_{j_r}\) we denoted the \(j_r\)-quantile of the semicircular law.
\end{lemma}

Second, we compare the processes \(\widetilde{\bm \mu}^{(r)}(t)\), \({\bm \mu}^{(r)}(t)\), i.e.\ we control the error made by replacing the weakly correlated Brownian motions \(\underline{\bm \zeta}^{\mathrm{in}}\) by the independent Brownian motions \(\underline{\bm \beta}^{\mathrm{in}}\).
\begin{lemma}\label{lem:mu}
   Let \({\bm \mu}^{(r)}(t)\), \(\widetilde{\bm \mu}^{(r)}(t)\), with \(r\in [2]\), be the processes defined in~\eqref{eq:indeppro} and~\eqref{eq:compmu}, respectively. For any small \(\omega_1,\zeta_1>0\) there exists \(\omega>0\), with \(\omega\ll \omega_1\), such that for any \(N^{-\zeta_1}\le |x_1-x_2|\le c^*\) it holds
   \begin{equation}
      \abs*{\delta\mu_{j_r}^{(r)}(t_1)-\delta\widetilde{\mu}_{j_r}^{(r)}(t_1)}\le N^{-1-\omega},
   \end{equation}
   with very high probability on the probability space \(\Omega_\beta\), where \(t_1:=N^{-1+\omega_1}\).
\end{lemma}
The key ingredient for the proof of Lemma~\ref{lem:mu} is the following fundamental bound on the eigenvector overlaps in~\eqref{eq:bover} proven in Section~\ref{sec:overlap}, which ensures that the correlation \(\Theta_{\mathfrak{i}\mathfrak{j}}^{x_r,x_m}\) in~\eqref{eq:thetaover} is small.
\begin{proposition}\label{pro:bover} Given \(c_0, c_1\) as in Proposition~\ref{pro:realmaintres},
   assume \(\braket{\mathring{A}^2}\ge c_0\), \(\| A\|\lesssim 1\).
   There exists \(c^*\) depending on \(c_0, c_1\)  such that the following holds
   for any  small \(\zeta_1>0\).
   Pick \(x_1, x_2\) such that \(N^{-\zeta_1}\le |x_1-x_2|\le c^*\), and
   let \(\{{\bm u}_i^{x_r}\}_{i\in [N]}\), for \(r\in [2]\), be the orthonormal eigenbasis of the matrices \(H+x_r A\).
   Then there exists \(\omega_E>0\) such that
   \begin{equation}\label{eq:bover}
      |\braket{{\bm u}_{j_1}^{x_1},{\bm u}_{j_2}^{x_2}}|\le N^{-\omega_E} %
   \end{equation}
   with very high probability for any \(j_1, j_2\) in the \(c_1\)-bulk  with \(|j_1-j_2|\lesssim N|x_1-x_2|\).
\end{proposition}
Using Lemmas~\ref{lem:lam}--\ref{lem:mu} as an input we conclude Proposition~\ref{pro:realmaintresa}.
\begin{proof}[Proof of Proposition~\ref{pro:realmaintresa}]
   By Lemma~\ref{lem:lam} we readily conclude that
   \begin{equation}\label{eq:firsappr}
      \begin{split}
         \E\left[P\left(N\delta\lambda_{j_1}^{x_1}(t_1)\right)Q\left(N\delta\lambda_{j_2}^{x_2}(t_1)\right)\right]&=\E\left[P\left(N\rho_1\delta\widetilde\lambda_{j_1}^{(1)}(t_1)\right)Q\left(N\rho_2\delta\widetilde\lambda_{j_2}^{(2)}(t_1)\right)\right]\\
         &\quad +
         \mathcal{O}\left(N^{-\omega} \norm{P}_{C^1}\norm{Q}_{C^1} \right),
      \end{split}
   \end{equation}
   where we denoted \(\rho_r:=\rho_{\mathrm{sc}}(\gamma_{j_r})/\rho^{x_r}(\gamma_{j_r}^{x_r})\) and used the uniform boundedness of \(\rho_\mathrm{sc},\rho^{x}\). Then, by~\eqref{eq:samejointdistr}, it follows that
   \begin{equation}
      \E\left[P\left(N\rho_1\delta\widetilde\lambda_{j_1}^{(1)}(t_1)\right)Q\left(N\rho_2\delta\widetilde\lambda_{j_2}^{(2)}(t_1)\right)\right]=\E\left[P\left(N\rho_1\delta\widetilde\mu_{j_1}^{(1)}(t_1)\right)Q\left(N\rho_2\delta\widetilde\mu_{j_2}^{(2)}(t_1)\right)\right].
   \end{equation}
   Moreover, by Lemma~\ref{lem:mu}, we have that
   \begin{equation}\label{eq:mumutilde}
      \begin{split}
         \E\left[P\left(N\rho_1\delta\widetilde\mu_{j_1}^{(1)}(t_1)\right)Q\left(N\rho_2\delta\widetilde\mu_{j_2}^{(2)}(t_1)\right)\right]&=\E\left[P\left(N\rho_1\delta\mu_{j_1}^{(1)}(t_1)\right)Q\left(N\rho_2\delta\mu_{j_2}^{(2)}(t_1)\right)\right]\\
         &\quad + \mathcal{O}\left(N^{-\omega} \norm{P}_{C^1}\norm{Q}_{C^1} \right).
      \end{split}
   \end{equation}
   Additionally, by the definition of the processes \({\bm \mu}^{(r)}(t)\) in~\eqref{eq:indeppro} it follows that \({\bm \mu}^{(1)}(t)\), \({\bm \mu}^{(2)}(t)\) are independent, and so that
   \begin{equation}\label{eq:indepmus}
      \E\left[P\left(N\rho_1\delta\mu_{j_1}^{(1)}(t_1)\right)Q\left(N\rho_2\delta\mu_{j_2}^{(2)}(t_1)\right)\right] = \E\left[P\left(N\rho_1\delta\mu_{j_1}^{(1)}(t_1)\right)\right]\E\left[Q\left(N\rho_2\delta\mu_{j_2}^{(2)}(t_1)\right)\right]
   \end{equation}

   Combining~\eqref{eq:firsappr}--\eqref{eq:indepmus} we get
   \begin{equation}\label{eq:expprodprodexp}
      \begin{split}
         \E\left[P\left(N\rho_1\delta\lambda_{j_1}^{x_1}(t_1)\right)Q\left(N\rho_2\delta\lambda_{j_2}^{x_2}(t_1)\right)\right]&=\E\left[P\left(N\rho_1\delta\mu_{j_1}^{(1)}(t_1)\right)\right]\E\left[Q\left(N\rho_2\delta\mu_{j_2}^{(2)}(t_1)\right)\right]\\
         &\quad + \mathcal{O}\left(N^{-\omega} \norm{P}_{C^1}\norm{Q}_{C^1} \right).
      \end{split}
   \end{equation}
   Proceeding similarly to~\eqref{eq:firsappr}--\eqref{eq:mumutilde}, but for \(\E P\) and \(\E Q\) separately, we also conclude that
   \begin{equation}\label{eq:lastappr}
      \begin{split}
         \E\left[P\left(N\delta\lambda_{j_1}^{x_1}(t_1)\right)\right]\E\left[Q\left(N\delta\lambda_{j_2}^{x_2}(t_1)\right)\right]&=\E\left[P\left(N\rho_1\delta\mu_{j_1}^{(1)}(t_1)\right)\right]\E\left[Q\left(N\rho_2\delta\mu_{j_2}^{(2)}(t_1)\right)\right]\\
         &\quad + \mathcal{O}\left(N^{-\omega} \norm{P}_{C^1}\norm{Q}_{C^1} \right).%
      \end{split}
   \end{equation}
   Finally, combining~\eqref{eq:expprodprodexp}--\eqref{eq:lastappr}, we conclude the proof of~\eqref{eq:factora}.
\end{proof}

Before concluding this section with the proof of Lemmas~\ref{lem:lam}--\ref{lem:mu}, in Proposition~\ref{pro:similCLT}
below we state the main technical result used in their proofs. The proofs of these lemmas rely on extending the homogenisation analysis of~\cite[Theorem 3.1]{MR3914908}
to two DBM processes with weakly coupled driving Brownian motions. We used a very  similar idea
in~\cite[Section 7.4]{1912.04100} for DBM processes for singular values. We now first present the general version of this idea before applying it  to prove Lemmas~\ref{lem:lam}--\ref{lem:mu}.

In Proposition~\ref{pro:similCLT} below we compare the evolution of two DBMs whose driving Brownian motions are nearly the same for indices close to a fixed index \(i_0\) and are independent for indices away from \(i_0\).
Proposition~\ref{pro:similCLT} is the counterpart of~\cite[Proposition 7.14]{1912.04100},
where a similar analysis is performed for DBMs describing the evolution of particles
satisfying  slightly different DBMs.

Define the processes \(s_i(t)\), \(r_i(t)\) to be the solution of
\begin{equation}\label{eq:spro}
   \dif s_i(t)=\sqrt{\frac{2}{N}}\dif \mathfrak{b}_i^s(t)+\frac{1}{2N}\sum_{j\ne i}\frac{1}{s_i(t)-s_j(t)}\dif t, \qquad i\in [N],
\end{equation}
and
\begin{equation}\label{eq:rpro}
   \dif r_i(t)=\sqrt{\frac{2}{N}}\dif \mathfrak{b}_i^r(t)+\frac{1}{2N}\sum_{j\ne i}\frac{1}{r_i(t)-r_j(t)}\dif t, \qquad i\in [N],
\end{equation}
with initial conditions \(s_i(0)=s_i\) being the eigenvalues of a deformed Wigner matrix \(H\) satisfying Assumption~\ref{ass:entr}, and \(r_i(0)=r_i\) being the eigenvalues of a GOE matrix. Here we used the same notations of~\cite[Eqs. (7.44)--(7.45)]{1912.04100} to make the comparison with~\cite{1912.04100} easier. For simplicity in~\eqref{eq:spro}--\eqref{eq:rpro} we consider the DBMs only in the real case (the complex case is completely analogous).

\begin{remark}
   In~\cite[Eqs. (7.44)]{1912.04100} we assumed that the initial condition
   \(s_i(0)=s_i\) were  general points  satisfying~\cite[Definition 7.12]{1912.04100}, and not necessary the singular values of a matrix. Here we choose \(s_i(0)=s_i\) to be the eigenvalues of a deformed Wigner matrix to make the presentation shorter and simpler, however Proposition~\ref{pro:similCLT} clearly holds also for collections of particles satisfying similar assumptions to~\cite[Definition 7.12]{1912.04100}.
\end{remark}

We now formulate the assumptions on the driving Brownian motions in~\eqref{eq:spro}--\eqref{eq:rpro}. Set an \(N\)-dependent parameter \(K=K_N:=N^{\omega_K}\), for some small fixed \(\omega_K>0\).

\begin{assumption}\label{ass:close}
   Suppose that the families \(\{\mathfrak{b}^s_i\}_{i\in [N]}\), \(\{\mathfrak{b}^r_i\}_{i\in [N]}\) in~\eqref{eq:spro} and~\eqref{eq:rpro} are realised on a common probability space. Let
   \begin{equation}\label{eq:defL}
      L_{ij}(t) \dif t:= \dif \bigl[\mathfrak{b}^s_i(t)-\mathfrak{b}^r_i(t), \mathfrak{b}^s_j(t)-\mathfrak{b}^r_j(t)\bigr]
   \end{equation}
   denote their quadratic covariation (in~\cite[Eqs. (7.46)]{1912.04100} we used a different notation to denote the covariation). Fix an index \(i_0\) in the bulk of \(H\), and let the processes satisfy the following assumptions:
   \begin{enumerate}[label=(\alph*)]
      \item\label{close1} \(\{ \mathfrak{b}^s_i\}_{i\in [N]}\), \(\{ \mathfrak{b}^r_i\}_{i\in [N]}\) are two families of i.i.d.\ standard real Brownian motions.
      \item\label{close2} \(\{ \mathfrak{b}^r_i\}_{|i-i_0|> K}\) is independent of \(\{\mathfrak{b}^s_i\}_{i=1}^N\),
      and \(\{ \mathfrak{b}^s_i\}_{|i-i_0|> K}\) is independent of \(\{\mathfrak{b}^r_i\}_{i=1}^B\).
      \item\label{close3} Fix \(\omega_Q>0\) so that \(\omega_K\ll \omega_Q\). We assume that the subfamilies \(\{\mathfrak{b}^s_i\}_{|i-i_0|\le K}\), \(\{\mathfrak{b}^r_i\}_{|i-i_0|\le K}\) are very strongly dependent  in the sense that for any \(\abs{i-i_0}, \abs{j-i_0}\le K\) it holds
      \begin{equation}\label{eq:assbqv}
         \abs{L_{ij}(t)}\le N^{-\omega_Q}
      \end{equation}
      with very high probability for any fixed \(t\ge 0\).
   \end{enumerate}
\end{assumption}

Let \(\rho\) denote the self-consistent density of \(H\), and recall that \(\rho_{\mathrm{sc}}\) denotes the semicircular
density. By \(\rho_t\), \(\rho_{\mathrm{sc},t}\) we denote the evolution of \(\rho\) and \(\rho_{\mathrm{sc}}\), respectively, along the semicircular flow (see e.g.~\cite[Eq. (4.1)]{MR4026551})
and let  \(\widehat{\gamma}_i(t)\), \(\gamma_i(t)\) denote the quantiles of \(\rho_t\) and \(\rho_{\mathrm{sc},t}\).

\begin{proposition}\label{pro:similCLT}
   Let the processes \({\bm s}(t)=\{s_i(t)\}_{i\in [N]}\), \({\bm r}(t)=\{r_i(t)\}_{i\in [N]}\) be the solutions of~\eqref{eq:spro} and~\eqref{eq:rpro},  and assume that the driving Brownian motions in~\eqref{eq:spro}--\eqref{eq:rpro} satisfy Assumption~\ref{ass:close}. Let \(i_0\) be the index fixed in Assumption~\ref{ass:close}. Then for any small
   \(\omega_1, \omega_\ell>0\) such that  \(\omega_1\ll \omega_\ell\ll \omega_K\ll \omega_Q\) there exist \(\omega,\widehat{\omega}>0\) with \(\widehat{\omega}\ll \omega\ll \omega_1\), and such that it holds
   \begin{align}\label{eq:hihihi}
       & \rho(\widehat{\gamma}_{i_0})[s_{i_0+i}(t_1)-\widehat{\gamma}_{i_0}(t_1)]-\rho_{\mathrm{sc}}(\gamma_{i_0})[r_{i_0+i}(t_1)-\gamma_{i_0}(t_1)]                                                                                                                         \\\nonumber
       & =\sum_{|j|\le N^{2\omega_1}}\frac{1}{N}p_{t_1}\left(0, \frac{-j}{N\rho_{\mathrm{sc}}(\gamma_{i_0})}\right)\big[\rho(\widehat{\gamma}_{i_0})(s_{i_0+j}^{x_r}(0)-\widehat{\gamma}_{i_0}(0))-\rho_{\mathrm{sc}}(\gamma_{i_0})(r_{i_0+j}^{(r)}(0)-\gamma_{i_0}(0))\big] \\\nonumber
       & \qquad +\mathcal{O}(N^{-1-\omega}),
   \end{align}
   for any \(|i|\le N^{\widehat{\omega}}\), with very high probability, where \(t_1:= n^{-1+\omega_1}\) and \(p_t(x,y)\) is the fundamental solution (heat kernel) of the parabolic equation
   \begin{equation}\label{heateq}
      \partial_t f(x)=\int_{|x-y|\le\eta_\ell}\frac{f(y)-f(x)}{(x-y)^2}\rho_{\mathrm{sc}}(\gamma_{i_0})\,\dif y,
   \end{equation}
   with \(\eta_\ell:= N^{-1+\omega_\ell}\rho_{\mathrm{sc}}(\gamma_{i_0})^{-1}\)
   (see~\cite[Eqs.~(3.88)--(3.89)]{MR3914908} for more details).
\end{proposition}

\begin{proof}
   The proof of~\eqref{eq:hihihi}  is nearly identical to that of~\cite[Theorem 3.1]{MR3914908}
   up to a straightforward modification owing to the fact  that
   the driving Brownian motions in~\eqref{eq:spro}--\eqref{eq:rpro} are not exactly
   identical but they are very strongly correlated, see~\eqref{eq:assbqv}.
   A similar modification to handle this strong correlation was explained in details in a closely related context
   in~\cite[Proof of Proposition 7.14 in Section 7.6]{1912.04100}, with the difference that
   in~\cite{1912.04100}
   singular values were
   considered instead of eigenvalues hence the corresponding DBMs are slightly different.
   Furthermore, Proposition~\ref{pro:similCLT} is stated in a simpler form than~\cite[Proposition 7.14]{1912.04100}
   since the initial conditions are already eigenvalues and not arbitrary points hence they automatically
   satisfy certain regularity assumptions. The precise changes due to this
   simplification  are described in the technical Remark~\ref{rm:dif}
   below.
\end{proof}

\begin{remark}\label{rm:dif}
   There a few differences in the setup of Proposition~\ref{pro:similCLT} and~\cite[Proposition 7.14]{1912.04100}. These are caused by the fact that we now consider \(s_i(0)=s_i\) to be the eigenvalues of a deformed Wigner matrix \(H\), instead of a collection of particles satisfying~\cite[Definition 7.12]{1912.04100}. In particular, \(\nu\) in~\cite[Definition 7.12]{1912.04100} can be chosen equal to zero, then, since the eigenvalues of \(H\) are regular (\cite[Eq. (7.48)]{1912.04100}) on an order one scale, we can choose \(g=N^{-1+\xi}\), for an arbitrary small \(\xi>0\), and \(G=1\) in~\cite[Definition 7.12]{1912.04100}. Additionally, \(t_f=N^{-1+\omega_f}\) is replaced by \(t_1=N^{-1+\omega_1}\), and \(\rho_{\mathrm{fc},t_f}\) in is replaced by \(\rho\). Finally, we remark that in~\cite[Proposition 7.14]{1912.04100} for \(\omega_f\) we required that \(\omega_K\ll \omega_f\ll \omega_Q\), instead in Proposition~\ref{pro:similCLT} we required that \(\omega_1\ll\omega_K\ll \omega_Q\). This discrepancy is caused by the fact that in the proof of~\cite[Proposition 7.14]{1912.04100} we first needed to run the DBM for \(s_i(t)\) for an initial time \(t_0=N^{-1+\omega_0}\) to regularise the particles \(s_i(0)=s_i\), with \(\omega_K\ll \omega_0\ll \omega_Q\), and then run both DBMs for an additional time \(N^{-1+\omega_1}\), with \(\omega_1\ll\omega_K\ll \omega_0\ll \omega_Q\) (see below~\cite[Eq. (7.56)]{1912.04100}). Finally, in~\cite[Proposition 7.14]{1912.04100} we have \(t_f:=t_0+t_1\sim t_0\), hence the reader can think \(\omega_f=\omega_0\). In the current case we do not need to run~\eqref{eq:spro} for an initial time \(t_0\) since \(s_i(0)=s_i\) are already regular being the eigenvalues of a deformed Wigner matrix.
\end{remark}

We are now ready to prove Lemmas~\ref{lem:lam}--\ref{lem:mu}.

\begin{proof}[Proof of Lemmas~\ref{lem:lam}--\ref{lem:mu}]
   By construction the processes \({\bm \lambda}^{x_r}(t)\), \(\widetilde{\bm \lambda}^{(r)}(t)\) satisfy the assumptions of Proposition~\ref{pro:similCLT} with \(i_0=j_r\), \(i=0\), \(\rho=\rho^{x_r}\)
   and \(\omega_K= \omega_A\).  Hence, by Proposition~\ref{pro:similCLT}, we get
   \begin{align}\label{eq:singleeighom}
       & \rho^{x_r}(\gamma_{j_r}^{x_r})[\lambda_{j_r}^{x_r}(t_1)-\gamma_{j_r}^{x_r}(t_1)]-\rho_{\mathrm{sc}}(\gamma_{j_r})[\widetilde{\lambda}_{j_r}^{(r)}(t_1)-\gamma_{j_r}(t_1)]                                                                                                                 \\\nonumber
       & =\sum_{|j|\le N^{2\omega_1}}\frac{1}{N}p_{t_1}\left(0, \frac{-j}{N\rho_{\mathrm{sc}}(\gamma_{j_r})}\right)\big[\rho^{x_r}(\gamma_{j_r}^{x_r})(\lambda_{j_r+j}^{x_r}(0)-\gamma_{j_r}^{x_r}(0))-\rho_{\mathrm{sc}}(\gamma_{j_r})(\widetilde{\lambda}_{j_r+j}^{(r)}(0)-\gamma_{j_r}(0))\big] \\\nonumber
       & \qquad +\mathcal{O}(N^{-1-\kappa}),
   \end{align}
   with very high probability, for some small fixed \(\kappa>0\). Here \(\gamma_{j_r}^{x_r}(t)\), \(\gamma_{j_r}(t)\) denote the quantiles of \(\rho_t^{x_r}\) and \(\rho_{\mathrm{sc},t}\), respectively, with \(\rho_t^{x_r}\), \(\rho_{\mathrm{sc},t}\) the evolution of \(\rho^{x_r}\), \(\rho_{\mathrm{sc}}\) along the semicircular flow (see e.g.~\cite[Eq.~(4.1)]{MR4026551})
   and \(p_t(x,y)\) is defined in~\eqref{heateq}.

   Analogously, we observe that the processes \(\widetilde{\bm\mu}^{(r)}(t)\), \({\bm \mu}^{(r)}(t)\) satisfy the assumptions of Proposition~\ref{pro:similCLT} with \(i_0=j_r\), \(i=0\), \(\rho=\rho_{\mathrm{sc}}\), \(\omega_K=\omega_A\), \(\omega_Q=\omega_E\)  due to~\eqref{eq:bover} (in particular~\eqref{eq:bover} is needed to check Assumption~\ref{ass:close}--\ref{close3}), and thus we obtain
   \begin{equation}\label{eq:singleeighommu}
      \mu_{j_r}^{(r)}(t_1)-\widetilde{\mu}_{j_r}^{(r)}(t_1)=\sum_{|j|\le N^{2\omega_1}}\frac{1}{N}p_{t_1}\left(0, \frac{-j}{N\rho_{\mathrm{sc}}(\gamma_{j_r})}\right)\big[\mu_{j_r+j}^{(r)}(0)-\widetilde{\mu}_{j_r+j}^{(r)}(0)\big]+\mathcal{O}(N^{-1-\kappa}).
   \end{equation}

   From now on we focus only on the precesses \({\bm \lambda}^{x_r}(t_1)\), \(\widetilde{\bm\lambda}^{(r)}(t_1)\) and so on the proof of Lemma~\ref{lem:lam}. The proof to conclude Lemma~\ref{lem:mu} is completely analogous and so omitted. Combining~\eqref{eq:singleeighom} with another application of Proposition~\ref{pro:similCLT}, this time for \(i_0=j_r\) and \(i=1\), we readily conclude that

   \begin{equation}\label{eq:hpa}
      \begin{split}
         &\rho^{x_r}(\gamma_{j_r}^{x_r})\big[\lambda_{j_r+1}^{x_r}(t_1)-\lambda_{j_r}^{x_r}(t_1)\big]-\rho_{\mathrm{sc}}(\gamma_{j_r})\big[\widetilde{\lambda}_{j_r+1}^{(r)}(t_1)-\widetilde{\lambda}_{j_r}^{(r)}(t_1)\big] \\
         &\qquad\quad=\sum_{|j|\le N^{2\omega_1}}\frac{1}{N}\left[p_{t_1}\left(0, \frac{1-j}{N\rho_{\mathrm{sc}}(\gamma_{j_r})}\right)-p_{t_1}\left(0, \frac{-j}{N\rho_{\mathrm{sc}}(\gamma_{j_r})}\right)\right] \\
         &\qquad\qquad\quad\times \big[\rho^{x_r}(\gamma_{j_r}^{x_r})(\lambda_{j_r+j}^{x_r}(0)-\gamma_{j_r}^{x_r})-\rho_{\mathrm{sc}}(\gamma_{j_r})(\widetilde{\lambda}_{j_r+j}^{(r)}(0)-\gamma_{j_r})\big]+\mathcal{O}(N^{-1-\kappa}) \\
         &\qquad\quad=\mathcal{O}(N^{-1-\kappa+\xi}+N^{-1-\omega_1+\xi}),
      \end{split}
   \end{equation}
   with very high probability, where we used rigidity
   \begin{equation}\label{eq:rig1}
      |\lambda_i^{x_r}-\gamma_i^{x_r}|\le \frac{N^\xi}{N},
   \end{equation}
   a similar rigidity bound for \(\widetilde{\lambda}_i^{(r)}\). Additionally, to go to the last line of~\eqref{eq:hpa} we used the following properties of the heat kernel \(p_{t_1}(x, y)\):
   \begin{equation} \label{eq:heatkernel}
      \begin{split}
         \left|p_{t_1}\left(0, \frac{1-j}{N\rho_{\mathrm{sc}}(\gamma_{j_r})}\right)-p_{t_1}\left(0, \frac{-j}{N\rho_{\mathrm{sc}}(\gamma_{j_r})}\right)\right|&\le\frac{1}{N\rho_{\mathrm{sc}}(\gamma_{j_r})}\int_0^1 \left|\partial_y p_{t_1}\left(0, \frac{\tau-j}{N\rho_{\mathrm{sc}}(0)}\right)\right| \dif \tau \\
         &\lesssim \frac{1}{N t_1}\int_0^1 p_{t_1}\left(0, \frac{\tau-j}{N\rho_{\mathrm{sc}}(\gamma_{j_r})}\right) \dif \tau \\
         \frac{1}{N}\sum_{|j|\le N^{2\omega_1}}p_{t_1}\left(0, \frac{\tau-j}{N\rho_{\mathrm{sc}}(0)}\right)&=1+\mathcal{O}(N^{-\omega_1}).
      \end{split}
   \end{equation}
   The bound in the second line of~\eqref{eq:heatkernel} follows by~\cite[Eq. (3.96)]{MR3914908}. The second relation of~\eqref{eq:heatkernel} follows by~\cite[Eqs. (3.90), (3.103)]{MR3914908}. The bound in~\eqref{eq:hpa} concludes the proof of Lemma~\ref{lem:lam}.
\end{proof}

\subsection{Proof of Proposition~\ref{pro:realmaintressam}}\label{sec:36}
We now turn to the proof of~Proposition~\ref{pro:realmaintressam}. We first present~\ref{case2} which is structurally very similar to the proof of Proposition~\ref{pro:realmaintressam}. Afterwards we turn to~\ref{case1} which is easier but additionally requires to modify the flow~\eqref{eq:matrou} to account for the correlations among entries of \(H\).
\subsubsection{\ref{case2}}
Proceeding as in~\eqref{eq:matrdbm}--\eqref{eq:imprel} and
using the notations and  assumptions from~\ref{case2} of Proposition~\ref{pro:realmaintressam}, it is enough to prove
\begin{equation}\label{eq:mainneweq}\Cov\left[P\left(N\delta\lambda_{j_1}^{x_1}(t_1)\right), Q\left(N\delta\lambda_{j_2}^{x_2}(t_1)\right)\right]=\mathcal{O}\left(N^{-\zeta_2} \| P\|_{C^1} \| Q\|_{C^1}\right),
\end{equation}
with \(t_1=N^{-1+\omega_1}\), for some small \(\omega_1>0\). Here \({\bm \lambda}(t)\) are the eigenvalues of \(\widehat{H}_t\), which is the solution of~\eqref{eq:matrdbm} with initial condition \(\widehat{H}_0=H^\#_{t_1}\), where \(H^\#_{t_1}\) is from in~\eqref{eq:solou}.

The proof of~\eqref{eq:mainneweq} follows by a DBM analysis very similar to the one in Section~\ref{sec:DBMDBM}. More precisely, all the processes \({\bm\lambda}^{x_r}(t)\), \(\widetilde{\bm\lambda}^{(r)}(t)\), \(\widetilde{\bm\mu}^{(r)}(t)\), and \({\bm\mu}^{(r)}(t)\) are defined exactly in the same way; the only difference is that Proposition~\ref{pro:bover} has to be replaced by the following bound on the eigenvector overlap (its proof will be given at the end
of Section~\ref{sec:overlap}).

\begin{proposition}\label{pro:bover1}  We are in the setup of~\ref{case2} of Proposition~\ref{pro:realmaintressam}.
   For any small \(c_1>0\)  there exists a \(c_0>0\)  and a \(c_*\) depending on \(c_0, c_1\)
   such that the following hold for any \(\zeta_1>0\) sufficiently small.
   Assume \(c_0|\braket{A}|\ge \braket{\mathring{A}^2}^{1/2}\),
   \(|\braket{A}|\ge c_0\),
   \(\| A\|\lesssim 1\). Pick indices \(j_1, j_2\) with \( N^{1-\zeta_1}\le |j_1-j_2|\le c_*N\)
   and choose \(x_1, x_2\)  with
   \(|x_1-x_2|\lesssim |j_1-j_2|/N\) such that \(\rho^{x_r}(\gamma_{j_r}^{x_r})\ge c_1\).
   Let \(\{{\bm u}_i^{x_r}\}_{i\in [N]}\), for \(r\in [2]\), be the orthonormal eigenbasis of the matrices \(H+x_r A\).
   Then there exists \(\omega_E>0\) such that
   \begin{equation}\label{eq:bover1}
      |\braket{{\bm u}_{j_1}^{x_1},{\bm u}_{j_2}^{x_2}}|\le N^{-\omega_E} %
   \end{equation}
   with very high probability.
\end{proposition}

Then, using~\eqref{eq:bover1}, instead of~\eqref{eq:bover}, as an input we readily conclude the analogous versions of Lemmas~\ref{lem:lam}--\ref{lem:mu}. Finally, by Lemmas~\ref{lem:lam}--\ref{lem:mu} we conclude the proof of~\eqref{eq:mainneweq} proceeding exactly as in~\eqref{eq:firsappr}--\eqref{eq:lastappr}.

\subsubsection{\ref{case1}}\label{app:samp} %

In this case we consider the following Ornstein-Uhlenbeck (OU) flow instead of~\eqref{eq:matrou}:
\begin{equation}\label{eq:corrou}
   \dif H_t=-\frac{1}{2}(H_t-\E H_0) \dif t+\frac{\Sigma^{1/2}[\dif B_t]}{\sqrt{N}}, \qquad H_0=H, \qquad \Sigma[\cdot]:=\frac{\beta}{2}\E W\Tr[W\cdot]
\end{equation}
Here \(W=\sqrt{N}(H- \E H)\) and note that
the OU flow is chosen to keep the expectation and the covariance structure of \(H_t\) invariant under the time evolution. As usual, the parameter
\(\beta=1\) in the real case and \(\beta=2\) in the complex case. Here \(\Sigma^{1/2}\) denotes
the square root of the positive operator \(\Sigma\) acting on \(N\times N\) matrices equipped
with the usual Hilbert-Schmidt scalar product structure.

Then proceeding as in~\eqref{eq:matrdbm}--\eqref{eq:imprel}, after replacing~\eqref{eq:matrou} with~\eqref{eq:corrou}, we find that to conclude the proof of this proposition it is enough to prove
\begin{equation}\label{eq:easyid}
   \Cov\left[P\left(N\delta\lambda_{j_1}^{x_1}(t_1)\right), Q\left(N\delta\lambda_{j_2}^{x_2}(t_1)\right)\right]=\mathcal{O}\left(N^{-\zeta_2} \| P\|_{C^1}  \| Q\|_{C^1}\right)
\end{equation}
with \(t_1=N^{-1+\omega_1}\), for some small \(\omega_1>0\). Here \({\bm \lambda}(t)\) are the eigenvalues of \(\widehat{H}_t\), which is the solution of~\eqref{eq:matrdbm} with initial condition \(\widehat{H}_0=H^\#_{t_1}\), where \(H^\#_{t_1}\) is from in~\eqref{eq:solou} with \(H_{t_1}\) coming from~\eqref{eq:corrou}.

Note that for \(A=I\) the gaps \(\lambda_{i+1}^x-\lambda_i^x=\lambda_{i+1}-\lambda_i\) do not depend on \(x\). In particular~\eqref{eq:easyid} simplifies since
\begin{equation}\label{eq:easyid1}
   \delta\lambda_{j_r}^{x_r}=\lambda_{j_r+1}-\lambda_{j_r},
\end{equation}
where we recall that \(\{\lambda_i\}_{i\in [N]}\) are the eigenvalues of \(H\), \(\rho\) is its limiting density of states, and \(\{\gamma_i\}_{i\in [N]}\) are the corresponding quantiles.

By~\eqref{eq:easyid1}, the proof of~\eqref{eq:easyid} is a much simpler version of the proof of Propositions~\ref{pro:realmaintresa} presented in Section~\ref{sec:DBMDBM} for general \(A\)'s.
More precisely, since for \(A=I\) the gaps are independent of \(x\), it is enough to consider the DBM for the evolution of the eigenvalues of \(H\) instead of \(H^x\):
\begin{equation}\label{eq:simpledbm}
   \dif\lambda_i(t)=\sqrt{\frac{2}{N}}\dif b_i(t)+\frac{1}{N}\sum_{j\ne i}\frac{1}{\lambda_i(t)-\lambda_j(t)}\dif t,
\end{equation}
with \(\{b_i\}_{i\in [N]}\) a family of standard i.i.d.\ real Brownian motions
(we wrote up the real symmetric case for simplicity).
The fact that
\begin{equation}\label{eq:trivqv}
   \dif [b_i(t),b_j(t)]=\delta_{ij}\dif t,
\end{equation}
follows by the orthogonality of the eigenvectors of \(H\). Note that~\eqref{eq:simpledbm} does not depend on \(x\), unlike~\eqref{eq:eigflow} in Section~\ref{sec:DBM}. In particular by~\eqref{eq:trivqv} it follows that \(C_{ij}(t)\equiv I\) in~\eqref{eq:defcov}; indeed Proposition~\ref{pro:bover} is trivially satisfied by orthogonality since \(j_1\) and \(j_2\) are sufficiently away from each other by assumption. Additionally, it is not necessary to define the comparison processes \(\widetilde{\bm \lambda}\), \(\widetilde{\bm \mu}\) since the driving Brownian motions in~\eqref{eq:simpledbm} are completely independent among each other, hence the processes \({\bm \lambda}(t)\)
with indices close to \(j_r\) and \({\bm \mu}^{(r)}(t)\) can be compared directly (see e.g.~\cite[Section 3]{MR3914908}).

\section{Bound on the eigenvector overlap}\label{sec:overlap}
The overlap in~\eqref{eq:bover} and in~\eqref{eq:bover1} will be estimated by a local law
involving the trace of the product of  the resolvents  of \(H^{x_1}\) and \(H^{x_2}\) for any fixed \(x_1, x_2\).
Individual resolvents can be approximated by the solution \(M\) of the MDE~\eqref{MDE}
but the deterministic approximation of products of resolvents are not simply products
of \(M\)'s.
Local laws are typically proven by deriving an approximate self-consistent equation
and then effectively  controlling its stability.
In Proposition~\ref{prop evec orth} we formulate a more accurate form of
the overlap bounds~\eqref{eq:bover}--\eqref{eq:bover1} in terms of the stability factor
of the self-consistent equation for the product of two resolvents.
In the subsequent Lemma~\ref{lem:stabop} we give an effective control on
this stability factor. Proposition~\ref{prop evec orth} will be proven in this section
while the proof of Lemma~\ref{lem:stabop} is postponed to Appendix~\ref{sec:M}.

For notational convenience we introduce the commonly used notion of \emph{stochastic domination}. For some family of non-negative random variables \(X=X(N)\ge 0\) and a deterministic control parameter \(\psi>0 \) we write \(X\prec \psi\) if for each \(\epsilon>0,D>0\) there exists some constant \(C\) such that \[\Prob(X>N^{\epsilon}\psi)\le CN^{-D}.\]
\begin{proposition}\label{prop evec orth}
   Let \(\{{\bm u}_i^{x_r}\}_{i\in [N]}\), for \(r=1,2\), be the orthonormal eigenbasis of the matrices \(H+x_r A\) and fix
   indices \(i_1,i_2\) in the bulk i.e.\ with \(\braket{\Im M^{x_r}(\gamma_{i_r}^{x_r})}\gtrsim 1\).
   Then it holds that\footnote{Star in bracket  \(M^{(*)}\) indicates that the statement holds for both \(M\) and its adjoint
      \(M^*\)}
   \begin{equation}\label{eq ov bound}
      \abs{\braket{\bm u_{i_1}^{x_1},\bm u_{i_2}^{x_2}}}^2 \prec N^{-1/15} \delta^{-16/15},\quad \delta:=\abs{1-\braket{M^{x_1}(\gamma_{i_1}^{x_1}) M^{x_2}(\gamma_{i_2}^{x_2})^{(\ast)} }}
   \end{equation}
   whenever \(N^{-1/6}\lesssim \delta\lesssim 1\).
\end{proposition}
\begin{lemma}\label{lem:stabop} For any  \(c_1>0\) there is a \(c^*\) such that
   for any \(x_1,x_2,E_1,E_2\) such that \(|x_1-x_2|+|E_1-E_2|\le c^*\)  and \(\rho^{x_r}(E_r)\ge c_1\), \(r=1,2\),
   it holds
   \begin{equation}\label{eq:lbstop}
      \begin{split}
         \abs{1-\braket{M^{x_1}(E_1)M^{x_2}(E_2)^{(*)}}}&\gtrsim |E_1-x_1\braket{A}-E_2+x_2\braket{A}|^2+|x_1-x_2|^2\braket{\mathring{A}^2}\\
         &\qquad +\mathcal{O}(|x_1-x_2|^3+|E_1-E_2|^3).
      \end{split}
   \end{equation}
\end{lemma}
For \(z_1,z_2\in\C\setminus\R\) we abbreviate
\begin{equation}
   G_i:=(H^{x_i}-z_i)^{-1}, \quad M_i:=M^{x_i}(z_i), \quad M_{ij} := \frac{M_i M_j}{1-\braket{M_j M_j}}
\end{equation}
and will prove the following \(G_1G_2\) local law.
\begin{proposition}\label{prop local law}
   Fix $\xi>0$ and let \(z_1,z_2\in\C\) with \(\abs{\Im z_1}=\abs{\Im z_2}=\eta\) be the bulk, i.e.\ \(\braket{\Im M_i}\gtrsim 1\), such that \(N\eta\delta_{12}\ge N^\xi\), where \(\delta_{12}:=\abs{\braket{1-\braket{M_1M_2}}}\). Then it holds that
   \begin{equation}
      \abs{\braket{G_1G_2 A-M_{12}A}} \prec \frac{\norm{A}}{\delta_{12} N\eta^2} \biggl(\eta^{1/12}+ \frac{1}{\sqrt{N\eta}} + \Bigl( \frac{\eta}{\delta_{12}} \Bigr)^{1/4} + \frac{1}{(\delta_{12} N\eta)^{1/3}} \biggr)
   \end{equation}
   uniformly in deterministic matrices \(A\).
\end{proposition}
\begin{proof}[Proof of Proposition~\ref{prop evec orth}]
   We will now apply Proposition~\ref{prop local law} with \(z_r=E_r\pm \ii\eta\), \(E_r:= \gamma_{i_r}^{x_r}\) and setting \(\eta:=(N\delta)^{-4/5}\). By 1/3-H\"older continuity of $z\mapsto M^x(z)$~\cite[Proposition 2.4]{MR4164728},
   we have \(\delta=\delta_{12}(1+\landauo{1})\) due to assumption~\(\delta\gtrsim N^{-1/6}\) and therefore the condition \(N\eta\delta_{12}\ge N^\xi\) of Proposition~\ref{prop local law} is fulfilled. Then, together with spectral decomposition of $\Im G_i$ we obtain
   \begin{equation}\label{eq:uu}
      \begin{split}
         &\sum_{\abs{\lambda_{j_1}^{x_1}-E_1}\lesssim \eta}\sum_{\abs{\lambda_{j_2}^{x_2}-E_2}\lesssim \eta} \abs{\braket{\bm u_{j_1}^{x_1},\bm u_{j_2}^{x_2}}}^2 \\
         &\quad\lesssim \sum_{j_1,j_2} \abs{\braket{\bm u_{j_1}^{x_1},\bm u_{j_2}^{x_2}}}^2 \frac{\eta^4}{[(\lambda_{j_1}^{x_1}-E_1)^2+\eta^2][(\lambda_{j_2}^{x_2}-E_2)^2+\eta^2]}\\
         &\quad = \eta^2  \Tr \Im G_1 \Im G_2 \prec N^{-1/16}\delta^{-16/15}.
      \end{split}
   \end{equation}
   By rigidity~\eqref{eq:rig1} the sums in the
   l.h.s.\ of~\eqref{eq:uu} contain the term \(\abs{\braket{\bm u_{i_1}^{x_1},\bm u_{i_2}^{x_2}}}^2\)
   as long as \(\eta \ge N^{-1+\xi}\). This relation clearly holds with our choice since \(\delta\lesssim 1\),
   concluding the proof.
\end{proof}
\begin{proof}[Proof of Proposition~\ref{prop local law}]
   This proof is an adaptation of a similar argument from~\cite[Theorem 5.2]{1912.04100}, so
   here we only give a short explanation.
   From~\eqref{MDE} obtain
   \begin{equation}
      \begin{split}
         (1-M_1M_2\braket{\cdot})[G_{1}G_2-M_{12}] = \Delta &:=-M_1 \un{W G_1 G_2} +  M_1 (G_2-M_2) \\
         &\quad + M_1\braket{G_1 G_2} (G_2-M_2) + M_1 \braket{G_1-M_1} G_1 G_2,
      \end{split}
   \end{equation}
   where
   \[\un{WG_1G_2}:=WG_1G_2+\braket{G_1}G_1G_2+\braket{G_1G_2}G_2.\]
   Thus we have
   \begin{equation}\label{eq G12}
      \braket{G_1 G_2A-M_{12}A}= \braket{\Delta A} + \frac{\braket{M_1M_2A}\braket{\Delta}}{1-\braket{M_1M_2}}.
   \end{equation}
   Recall that it was proven in~\cite[Proposition 5.3]{1912.04100} that if \(\abs{\braket{G_1G_2A}}\prec\norm{A}\theta\) for some constant  \(\theta\le\eta^{-1}\) uniformly in \(A\), then also
   \begin{equation}\label{eq WG12}
      \abs{\braket{\un{WG_1G_2A}}}\prec \frac{1}{N\eta^2} \Bigl((\theta\eta)^{1/4}+\frac{1}{(N\eta)^{1/2}}+\eta^{1/12}\Bigr),
   \end{equation}
   again uniformly in \(A\). Strictly speaking~\cite[Proposition 5.3]{1912.04100} was stated in the context of Hermitized i.i.d.\ random matrices. However, a simpler version of the same proof clearly applies to deformed Wigner matrices. The main simplification compared to~\cite{1912.04100} is that due to the constant variance profile of Wigner matrices summations as the one in~\cite[Eq.~(5.28a)]{1912.04100} can be directly performed, without introducing the block matrices \(E_1,E_2\). The remainder of the proof apart from the simplified \emph{resummation step} verbatim applies to the present case. Using~\eqref{eq WG12} in~\eqref{eq G12} and \(\theta\le\eta^{-1},\eta\lesssim1\) it follows that
   \begin{equation}\label{eq G12 est}
      \begin{split}
         \abs{\braket{G_1 G_2A-M_{12}A}} &\prec \frac{1}{\delta_{12}} \Bigl(\frac{1}{N\eta}+ \frac{\theta}{N\eta}+ \frac{1}{N\eta^2}\Bigl( \eta^{1/12} + \frac{1}{(N\eta)^{1/2}} + (\theta\eta)^{1/4} \Bigr) \Bigr)\\
         &\lesssim  \frac{1}{\delta_{12} N\eta^2}\Bigl( \eta^{1/12} + \frac{1}{(N\eta)^{1/2}} + (\theta\eta)^{1/4} \Bigr)
      \end{split}
   \end{equation}
   and therefore
   \begin{equation}\label{iter}
      \abs{\braket{G_1G_2A}} \prec \theta' := \frac{1}{\delta_{12}} + \frac{1}{\delta_{12} N\eta^2}\Bigl( \eta^{1/12} + \frac{1}{(N\eta)^{1/2}} + (\theta\eta)^{1/4} \Bigr).
   \end{equation}
   We now iterate~\eqref{iter} using that \(N\delta_{12}\eta\ge N^\xi\) starting from \(\theta_0=1/\eta\) (which follows trivially from Cauchy-Schwarz). In doing so we obtain a decreasing sequence of \(\theta\)'s and after finally many steps conclude that
   \begin{equation}
      \abs{\braket{G_1G_2 A}} \prec \theta_\ast,
   \end{equation}
   where \(\theta_\ast\) is the unique positive solution to the equation
   \begin{equation}
      \theta_\ast = \frac{1}{\delta_{12}} \Bigl( 1 + \frac{1}{N\eta^2}\Bigl( \eta^{1/12} + \frac{1}{(N\eta)^{1/2}}  \Bigr) \Bigr) + \frac{\theta_\ast^{1/4}}{\delta_{12} N\eta^{7/4}}.
   \end{equation}
   Asymptotically we have
   \begin{equation}
      \theta_\ast \sim \frac{1}{\delta_{12}} \Bigl( 1 + \frac{1}{N\eta^2}\Bigl( \eta^{1/12} + \frac{1}{(N\eta)^{1/2}}+\frac{1}{(\delta_{12} N\eta)^{1/3}}  \Bigr) \Bigr)
   \end{equation}
   and using~\eqref{eq G12 est} once more with \(\theta_\ast\) concludes the proof.
\end{proof}

\subsection{Proof of Propositions~\ref{pro:bover} and~\ref{pro:bover1}.}
Both proofs rely on Proposition~\ref{prop evec orth}
and proving that the lower bound
on the stability factor given in  Lemma~\ref{lem:stabop} with \(E_r=\gamma_{i_r}^{x_r}\), \(r=1,2\),
is bounded from below
by \(N^{-\epsilon}\)  with some small \(\epsilon\). This will be done separately
for the two propositions.

For Proposition~\ref{pro:bover} we use that \(|E_1-E_2| \lesssim |x_1-x_2|\le c^*\) with a small \(c^*\) and that
\(\braket{\mathring{A}^2}\gtrsim 1\), hence
\[
   \abs{1-\braket{M^{x_1}(\gamma_{i_1}^{x_1}) M^{x_2}(\gamma_{i_2}^{x_2})^{(\ast)} }}\gtrsim |x_1-x_2|^2\braket{\mathring{A}^2} \gtrsim N^{-2\zeta_1}.
\]
The relation \(|E_1-E_2| \lesssim |x_1-x_2|\) follows from
\[
   |E_1-E_2| = | \gamma_{i_1}^{x_1} -\gamma_{i_2}^{x_2} | \le | \gamma_{i_1}^{x_1} -\gamma_{i_1}^{x_2} |
   + |\gamma_{i_1}^{x_2}  - \gamma_{i_2}^{x_2} | \lesssim |x_1-x_2| + |i_1-i_2|/N
\]
and the fact  that \(|i_1-i_2|/N\lesssim |x_1-x_2|\) from the conditions of Propositions~\ref{pro:bover}.
The estimate on \(| \gamma_{i_1}^{x_1} -\gamma_{i_1}^{x_2} |\)  comes from Lemma~\ref{lem:gamma}.

For Proposition~\ref{pro:bover1} we have
\begin{equation}\label{11}
   \begin{split}
      |E_1-x_1\braket{A}-E_2+x_2\braket{A}| &\ge |\gamma^{x_1}_{i_1} - \gamma^{x_1}_{i_2} |
      - |\gamma^{x_1}_{i_2} - \gamma^{x_2}_{i_2} - (x_1-x_2) \braket{A}|
      \\
      &\ge c_1 |i_1-i_2|/N - C |x_1-x_2| \Big(\braket{\mathring{A}^2}^{1/2} + |x_1-x_2|\Big).
   \end{split}
\end{equation}
In estimating the first term we used that \(\gamma^{x_1}_{i_1}\), \(\gamma^{x_1}_{i_2}\)
are in the bulk, while we used~\eqref{eq:approxgamma} for  the second term.
Notice that
\[
   C |x_1-x_2| \Big(\braket{\mathring{A}^2}^{1/2} + |x_1-x_2|\Big) \le c_0\|A\| N^{-\zeta_1}
\]
by the bound \(\braket{\mathring{A}^2}^{1/2}\le c_0 |\braket{A}|\le c_0\| A\| \). Choosing \(c_0\)
sufficiently small, depending on \(c_1\), and recalling that \(|i_1-i_2|\ge N^{1-\zeta_1}\),
we can achieve that
\[
   C |x_1-x_2| \Big(\braket{\mathring{A}^2}^{1/2} + |x_1-x_2|\Big) \le \frac{1}{2} c_1 |i_1-i_2|/N
\]
in particular
\begin{equation}
   |E_1-x_1\braket{A}-E_2+x_2\braket{A}| \ge \frac{1}{2} c_1 |i_1-i_2|/N \gtrsim N^{-\zeta_1}
\end{equation}
from~\eqref{11}. This  shows the required lower bound for  the  leading  (first) term in~\eqref{eq:lbstop}.
The second term is non-negative.
The first error term is negligible,  \( |x_1-x_2|^3 \le  N^{-3\zeta_1}\). For the second error term we have
\[
   |E_1-E_2|\le  |\gamma_{i_1}^{x_2}  - \gamma_{i_2}^{x_2} |+
   | \gamma_{i_1}^{x_1} -\gamma_{i_1}^{x_2} |\lesssim |i_1-i_2|/N + |x_1-x_2| \lesssim  |i_1-i_2|/N
\]
using the upper bound on the density \(\rho^{x_2}\) in the first term and~\eqref{eq:approxgamma} in the second term. In the last step we used \(|x_1-x_2|\lesssim
|i_1-i_2|/N\) from the conditions of Proposition~\ref{pro:bover1}. This shows that the
error term \( |E_1-E_2|^3\lesssim (|i_1-i_2|/N)^3\) is negligible compared with the main term~\ref{12}
of order at least \((|i_1-i_2|/N)^2\) since we also assumed \(|i_1-i_2|/N\le c_*\) which is small.

This proves that
\[
   \abs{1-\braket{M^{x_1}(\gamma_{i_1}^{x_1}) M^{x_2}(\gamma_{i_2}^{x_2})^{(\ast)} }}\gtrsim
   |E_1-x_1\braket{A}-E_2+x_2\braket{A}|^2 \gtrsim  N^{-2\zeta_1}.
\]
in the setup of Proposition~\ref{pro:bover1} as well.

\section{Multi-gap quenched universality}\label{sec multi gap}
The following results are the multi-gap
versions of Theorems~\ref{thm1} and~\ref{thm2b}. The gaps will be tested by functions of \(k\) variables, so we define the set
\begin{equation}\label{eq:testfspac}
   \mathcal{F}_k=\mathcal{F}_{k,L,B}:=\set*{F\colon\R^k\to \R\given\supp (F)\subset [0,L]^k, \norm{F}_{C^5}\le B }
\end{equation}
of \(k\)-times differentiable and compactly supported test functions \(F\)
with some large constants \(L, B>0\).  In the following we will often use the notation \({\bm i}:=(i_1,\dots, i_k)\) for
a \(k\)-tuple of integer indices \(i_1,\dots, i_k\).
The gap distribution for \(H^x\) will be compared with that of the Gaussian Wigner matrices,
we therefore  let \(\{\mu_i\}_{i\in [N]}\) denote
the eigenvalues of a GOE/GUE matrix corresponding to the symmetry class of \(H\).

\begin{theorem}[Quenched universality via eigenbasis rotation mechanism]\label{thmB}
   Under the conditions of Theorem~\ref{thm1} for any \emph{\(c_1\)-bulk-index} \(i_0\) we have the following multi-gap version of Wigner-Dyson universality. There exist  \(\epsilon=\epsilon(a_0, c_0, c_1)>0\) and an event \(\Omega_{i_0,A}\) %
   with \(\Prob_H(\Omega_{i_0,A}^c)\le N^{-\epsilon}\)  such that  for all \(H\in \Omega_{i_0,A}\)
   the matrix \(H^x=H+xA\) satisfies

   \begin{equation}\label{eq:mainstat2new}
      \begin{split}
         &\max_{\norm{\bm i}_\infty\le K}\sup_{F\in\mathcal{F}_k}\Bigg|\E_x F\left(\left(N\rho^x(\gamma^x_{i_0})\delta\lambda_{i_0+i_j}^x%
         \right)_{j\in [k]}\right)-\E_\mu F\left(\left(N\rho_{\mathrm{sc}}(0)\delta\mu_{N/2+i_j}\right)_{j\in [k]}\right)\Bigg|\le CN^{-c},
      \end{split}
   \end{equation} %
   for \(K:=N^\zeta\) and some \(\zeta=\zeta(a_0,c_0,c_1)>0\), and \(c=c(k)>0\). The constant \(C\) in~\eqref{eq:mainstat2new} may depend on \(k,L, B, a_0, c_0, c_1\) and all constants in Assumptions~\ref{ass:entr} and~\ref{ass:x}  at most polynomially, but it is independent of \(N\).
\end{theorem}

\begin{theorem}[Quenched universality via spectral sampling mechanism]\label{thmA}
   Under the conditions of Theorem~\ref{thm2b} for any \emph{\(c_1\)-bulk-energy} \(E\) we have the following multi-gap version of Wigner-Dyson universality. There exists \(\epsilon=\epsilon(a_0, c_0, c_1)>0\) and an event \(\Omega_{E,A}\) with \(\Prob(\Omega_{E,A}^c)\le N^{-\epsilon}\) such that for all \(H\in \Omega_{E,A}\) the matrix \(H^x\) satisfies
   \begin{equation}\label{eq:mainstat3}
      \begin{split}
         &\max_{\norm{\bm i}_\infty\le K}\sup_{F\in\mathcal{F}_k}\Bigg|\E_x F\left(\left(N\rho^x(E)\delta\lambda^x_{i_0(x,E)+i_j}\right)_{j\in [k]}\right)-\E_\mu F\left(\left(N\rho_{\mathrm{sc}}(0)\delta\mu_{N/2+i_j}\right)_{j\in [k]}\right)\Bigg|\le CN^{-c},
      \end{split}
   \end{equation}
   where %
   \(K:=N^\zeta\), and some \(\zeta=\zeta(a_0,c_0,c_1)>0\), \(c=c(k)>0\).
   The constant \(C\) in~\eqref{eq:mainstat3} may depend on \(k,L, B, a_0, c_0, c_1\) and all constants in Assumptions~\ref{ass:entr} and~\ref{ass:x}  at most polynomially, but it is independent of \(N\).

\end{theorem}
First, to handle the supremum over the uncountable family \(\mathcal{F}_{k,L,B}\) of test functions \(F\) we reduce the problem to a finite set of test functions so that the union bound can be taken. Notice that for sufficiently smooth test functions \(F\), which are compactly supported on some box \([0,L]^k\) of size
\(L\), we can expand \(F\) in partial Fourier series as (see e.g.~\cite[Remark 3]{MR3306005} and~\cite[Eq. (30)]{MR4221653})
\begin{equation}\label{eq:fourappr}
   \begin{split}
      &F(x_1,\dots, x_k)=\sum_{|n_1|,\dots, |n_k|\in [N^{\zeta_*}]} C_F(n_1,\dots, n_k) \prod_{j=1}^k e^{ \ii n_j x_j/L}\varphi(x_j)+\mathcal{O}\left(N^{-c(\zeta_*)}\right), \\
      &\sum_{n_1,\dots, n_k} |C_F(n_1,\dots, n_k)|\lesssim 1,
   \end{split}
\end{equation}
with integer \(n_1,\dots,n_k\), where \(\varphi\colon\R\to \R\) is a smooth cut-off function such that it is equal to one on \([0,L]\) and it is equal to zero on \([-L/2,3L/2]^c\). Here \(\zeta_*>0\) is a small fixed constant that will be chosen later.
Introduce the notation
\begin{equation}\label{eq:newtestfunc}
   f_n(x):=e^{ \ii n x/L}\varphi(x),\qquad n\in\Z.
\end{equation}

\begin{proof}[Proof of Theorem~\ref{thmB}]
   By~\eqref{eq:fourappr} we get that
   \begin{align}\label{eq:secappr}
       & \sup_{\norm{\bm i}_\infty\le N^\zeta}\sup_F\Bigg|\E_x F\Bigl(\bigl(N\rho^x(\gamma_{i_0}^x)\delta\lambda_{i_0+i_j}^x\bigr)_{j\in [k]}\Bigr)-\E_\mu F\left(\left(N\rho_{\mathrm{sc}}(0)\delta\mu_{N/2+i_j}\right)_{j\in [k]}\right)\Bigg| \\\nonumber
       & \lesssim \sup_{\substack{\norm{\bm i}_\infty\le N^\zeta,                                                                                                                                                                                \\ \norm{\bm n}_{\infty}\le N^{\zeta_*}}}\left|\E_x\prod_{j\in [k]} f_{n_j}\left(N\rho^x(\gamma_{i_0}^x)\delta\lambda_{i_0+i_j}^x\right)-\E_\mu \prod_{j\in [k]}f_{n_j}\bigl(N\rho_{\mathrm{sc}}(0)\delta\mu_{N/2+i_j}\bigr)\right|+\mathcal{O}\left(N^{-c(\zeta_*)}\right),
   \end{align}
   with \(f_{n_j}\) defined in~\eqref{eq:newtestfunc} and \({\bm n}:=(n_1,\dots,n_k)\).

   Proceeding exactly as in the proof of Theorem~\ref{thm1} in Section~\ref{sec:univrot}, and using the fact that~\eqref{eq:neednow} holds for test functions \(P_1,P_2\) of \(k\) variables (see Remark~\ref{rem:newrem}), we conclude that for any fixed \(i_1,\dots, i_k\) and \(n_1,\dots, n_k\) there exists a probability event \(\Omega_{i_0,{\bm i}, {\bm n}}\), with \(\Prob(\Omega_{i_0,{\bm i}, {\bm n}}^c)\le N^{-\kappa}\), on which
   \begin{equation}\label{eq:lastbb}
      \left|\E_x\prod_{j\in [k]} f_{n_j}\left(N\rho^x(\gamma_{i_0}^x)\delta\lambda_{i_0+i_j}^x\right)-\E_\mu \prod_{j\in [k]}f_{n_j}\bigl(N\rho_{\mathrm{sc}}(0)\delta\mu_{N/2+i_j}\bigr)\right|\lesssim N^{-\alpha}\prod_{j\in [k]}\norm{f_{n_j}}_{C^5}.
   \end{equation}
   Then choosing \(\zeta,\zeta_*\le \kappa(10 k)^{-1}\) we define the event
   \begin{equation}\label{eq:hugeunb}
      \Omega_{i_0}:=\bigcap_{\norm{\bm i}\le N^\zeta}\bigcap_{\norm{\bm n}\le N^{\zeta_*}}\Omega_{i_0,{\bm i}, {\bm n}}, \qquad \Prob_H\left(\Omega_{i_0}^c\right)\lesssim N^{-\kappa}N^{k(\zeta+\zeta_*)}\le N^{-\kappa/2}.
   \end{equation}

   Finally, by~\eqref{eq:lastbb}--\eqref{eq:hugeunb}, for all \(H\in \Omega_{i_0}\), choosing \(\zeta_*\le \alpha (10k)^{-1}\),
   the claim~\eqref{eq:mainstat2new} follows with exponent
   \(c=\min\{\alpha-5k\zeta_*,c(\zeta_*)\}\) using that
   \(\norm{f_{n_j}}_{C^5} \le N^{5\zeta_*}\), for any \(j\in [k]\), and where \(c(\zeta_*)\) is from~\eqref{eq:fourappr}.
\end{proof}

\begin{proof}[Proof of Theorem~\ref{thmA}]
   Given~\eqref{eq:secappr}, the proof of Theorem~\ref{thmA}, following Section~\ref{sec:univsamp} instead of Section~\ref{sec:univrot} and using that Proposition~\ref{pro:realmaintressam} holds for \(P_1,P_2\) of \(k\) variables (see Remark~\ref{rem:newrem}), is completely analogous and so omitted.
\end{proof}

\appendix%
\section{Bound for the stability operator}\label{sec:M}
\begin{proof}[Proof of Lemma~\ref{lem:stabop}]
   Note that
   \[
      |1-\braket{M_1M_2^*}|\ge\Re[1-\braket{M_1M_2^*}]=\frac{1}{2}\braket{(M_1-M_2)(M_1-M_2)^*}+\mathcal{O}(\eta),
   \]
   where we used that \(\braket{M_i M_i^*}=1+\mathcal{O}(\eta)\), which follows by taking the imaginary in the MDE~\eqref{MDE}. Then using Taylor expansion in the \(x_2\) and the \(E_2\) variables we get that
   \begin{equation}\label{eq:impm12exp}
      \begin{split}
         &\braket{(M_1-M_2)(M_1-M_2)^*} \\
         &\qquad=\braket{(\partial_{x_1}M_1(x_2-x_1)+\partial_{E_1}M_1(E_2-E_1))(\partial_{x_1}M_1(x_2-x_1)+\partial_{E_1}M_1(E_2-E_1))^*} \\
         &\qquad\quad+\mathcal{O}(|x_1-x_2|^3+|E_1-E_2|^3).
      \end{split}
   \end{equation}
   To estimate the error term in~\eqref{eq:impm12exp} we
   used the following bounds for \(E=\Re z\)  in the bulk of the spectrum
   for any \(x\in [x_1, x_2]\) and \(E\in [E_1, E_2]\), a condition that is guaranteed by  \(|x_1-x_2|+ |E_1-E_2|\le c^*\) is small.
   By~\cite[Corollary 5.3, Lemma 5.7]{MR4164728} we have
   \begin{equation}\label{eq:boundneed}
      \left\lVert\partial_x^\alpha\partial_E^\beta M^x(E+\ii\eta)\right\rVert\le C_{\alpha, \beta}, \qquad
      \left\lVert\frac{1}{1-(M^x(z))^2\braket{\cdot}} \right\rVert_{\| \cdot\|\to \|\cdot \|} \le \frac{C}{\rho^x(z)[ \rho^x(z)
               +|\sigma^x(z)|]} ,
   \end{equation}
   for any \(\alpha,\beta\in\N\), for any fixed \(x\), where
   \(|\sigma^x(z)|\ge c\) unless \(\rho^x(z)\) has a near-cusp singularity and \(E=\Re z\) is close to this cusp
   point.
   Recall that the norm \(\norm{\cdot}\)  denotes the standard euclidean matrix norm on \(N\times N\) matrices. Here \(1-M^x(z)\braket{\cdot}M^x(z)\) is a linear operator acting on such matrices \(R\) as \((1-M^x(z)\braket{\cdot}M^x(z))[R] = R -M^x(z)\braket{R}M^x(z)\). Finally, the  second formula in~\eqref{eq:boundneed} involves the norm induced by the euclidean matrix norm.

   Then differentiating  the MDE in \(x\) and \(E\) we find that
   \[
      \partial_{x_1}M_1=-\frac{1}{1-M_1^2\braket{\cdot}}[M_1AM_1],
      \qquad \partial_{E_1}M_1=\frac{1}{1-M_1^2\braket{\cdot}}[M_1M_1].
   \]
   Hence, by
   \[
      \left(\frac{1}{1-M_1^2\braket{\cdot}}\right)\left(\frac{1}{1-M_1^2\braket{\cdot}}\right)^*\ge c, \qquad M_1M_1^*\ge c,
   \]
   we conclude
   \[
      \begin{split}
         &\braket{(\partial_{x_1}M_1(x_2-x_1)+\partial_{E_1}M_1(E_2-E_1))(\partial_{x_1}M_1(x_2-x_1)+\partial_{E_1}M_1(E_2-E_1))^*} \\
         &=\braket*{\abs*{\frac{1}{1-M_1^2\braket{\cdot}}M_1\big[-(x_1-x_2)A+(E_1-E_2)\big] M_1}^2} \\
         &\gtrsim \braket{[(E_1-E_2)-(x_1-x_2)A]^2}= |E_1-x_1\braket{A}-E_2+x_2\braket{A}|^2+
         |x_1-x_2|^2\braket{\mathring{A}^2},
      \end{split}
   \]
   where in the last equality we wrote \(A=\braket{A}+\mathring{A}\). This concludes the proof of~\eqref{eq:lbstop}
   in case when the adjoint is present. The estimate of  \(|1-\braket{M_1M_2}|\) is much easier, it
   follows directly from~\eqref{eq:boundneed}.
\end{proof}

\begin{proof}[Proof of Lemma~\ref{lem:gamma}]
   To make the presentation clearer we just consider the case \(x_1=x\) and \(x_2=0\), the general case is analogous and so omitted. For any fixed real parameters \(x, y\) consider the MDE
   \begin{equation}\label{eq:xymde}
      M^{-1}=z+B+x\braket{A}+y\mathring{A}+\braket{M}, \qquad \Im M\Im z>0.
   \end{equation}
   Note that for \(y=0\)~\eqref{eq:xymde} is the MDE for \(H\) and for \(y=x\)~\eqref{eq:xymde} is the one for \(H^x=H+xA\). We denote the unique solution of~\eqref{eq:xymde} by
   \(M^{x,y}=M^{x,y}(z)\),  the associated scDos by \(\rho^{x,y}\) and the corresponding quantiles  by \(\gamma_i^{x,y}\).
   We will use that
   \begin{equation}\label{12}
      \gamma_i^{x_1} - \gamma_i^{x_2} = \gamma_i^{x} - \gamma_i^{0} = \int_0^x \partial_s\gamma_i^{s,s} \dif s=
      \int_0^x \big[ \partial_x\gamma_i^{x,s}\big|_{x=s} + \partial_y\gamma_i^{s,y}\big|_{y=s} \big] \dif s.
   \end{equation}
   For the first term we use that that \(\partial_x\gamma_i^{x,s}=\braket{A}\), giving the leading term  \(x\braket{A}\)
   in Lemma~\ref{lem:gamma}.
   To estimate \(\partial_y\gamma_i^{s,y}\), we
   differentiate the defining equation of the quantiles
   \[
      \int_{-\infty}^{\gamma_i^{s,y}}\braket{\Im M^{s,y}(E)}\,\dif E =\frac{i}{N}
   \]
   with respect to \(y\).  We obtain
   \[
      \partial_y\gamma_i^{s,y}\braket{\Im M^{s,y}(\gamma_i^{s,y})}+\int_{-\infty}^{\gamma_i^{s,y}}\partial_y\braket{\Im M^{s,y}(E)}\,\dif E=0
   \]
   for any \(s,y\in [0,x]\).
   Then, using that in the bulk \(|\braket{\Im M^{s,y} (\gamma_i^{s,y})}|\ge c\), we conclude
   \begin{equation}\label{eq:bounhs}
      |\partial_y\gamma_i^{s,y}|\lesssim \int_{-C}^{\gamma_i^{s,y}}\left\langle\frac{1}{1-(M^{s,y}(E))^2\braket{\cdot}}M^{s,y}(E)\mathring{A}M^{s,y}(E)\right\rangle\,\dif E\lesssim \braket{\mathring{A}^2}^{1/2},
   \end{equation}
   where we used Schwarz inequality and the bounds in~\eqref{eq:boundneed}.
   The important fact about   the second bound in~\eqref{eq:boundneed} is that it is integrable in \(E\)
   since it has a \(|E-E_0|^{-1/2}\) singularity near an edge point
   \(E_0\) and a \(|E-E_0|^{-2/3}\) singularity near a cusp point  \(E_0\).
   Here we  also used that  \(|x|= |x_1-x_2|\le c^*\) is sufficiently small so that
   \(\gamma_i^{s,y}\)  is in the bulk not only  for \(s=y=x\), but for all \(s, y\in [0,x]\).
   From~\eqref{12} and~\eqref{eq:bounhs} we readily conclude~\eqref{eq:approxgamma}.
\end{proof}

\section{Numerics}\label{appendix figure}
Here we present numerical evidence quantifying the speed of convergence of the single gap distribution to its theoretical limit for the monoparametric ensemble, cf.\ Figure~\ref{KS}.
This numerics was inspired by the observation made in~\cite{GPSW}\footnote{We thank Stephen Shenker for communicating preliminary numerical results supporting this observation in June 2021.}
on the slow convergence
of the spectral form factor.
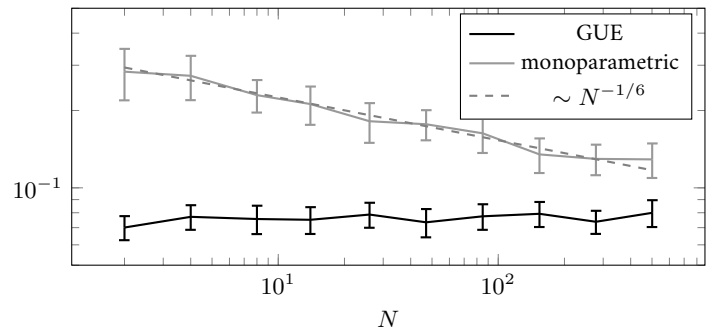
\begin{figure}
   \centering
   \begin{tikzpicture}
      \begin{axis}[width=10cm,height=5cm,xmode=log,ymode=log,ytick={.01,.1,1},ymin=0.05,ymax=.5,xlabel={\(N\)}]
         \addplot[thick,error bars/.cd, y dir = both, y explicit,error bar style={thick},error mark options={rotate=90, thick} ] table[col sep=comma,x index=0,y index=1,y error index=2] {KS.csv};
         \addplot[thick,draw=black!40,error bars/.cd, y dir = both, y explicit,error bar style={thick,black!40},error mark options={rotate=90, thick,black!40} ] table[col sep=comma,x index=0,y index=3,y error index=4] {KS.csv};
         \addplot [thick,dashed, domain=2:500, samples=19,gray,dashed] {.33*x^(-1/6)};
         \addlegendentry{GUE};
         \addlegendentry{monoparametric};
         \addlegendentry{\(\sim N^{-1/6}\)};
      \end{axis}
   \end{tikzpicture}
   \caption{The figure shows the Kolmogorov-Smirnov distance \(D(F,F'):=\sup_s\abs{F(s)-F'(s)}\) of the empirical cumulative distribution function (CDF) of the (rescaled) eigenvalue gap \(\lambda_{N/2+1}-\lambda_{N/2}\) to the CDF \(F_2\) corresponding to \(p_2\) for various values of \(N\) for both GUE and the monoparametric ensemble. The empirical CDF for the GUE has been generated by sampling \(100\) GUE matrices \(H\). For the monoparametric ensemble \(H^x=H+xA\)  typical GUE random matrices \(H,A \)
      have been fixed and \(100\) Gaussian random variables \(x\) have been sampled. The error bars represent the standard deviation of the obtained Kolmogorov-Smirnov distance for 50 independent repetitions. In accordance with Figure~\ref{figure quenched2} we find that the gap distribution for GUE matches its theoretical limit very well for any value of \(N\), while for the monoparametric ensemble the KS-distance seems to decay only slowly with \(N\).}\label{KS}
\end{figure}



\begin{thebibliography}{99} {}
   \bibitem{MR3916109}
   O. H. Ajanki, L. Erd\H{o}s, and T. Kr\"{u}ger, \emph{Stability of the matrix {D}yson equation and random matrices with correlations}, \href {https://doi.org/10.1007/s00440-018-0835-z} {Probab. Theory Related Fields \textbf{173}, 293–373 (2019)}, \href {http://www.ams.org/mathscinet-getitem?mr=3916109} {\nolinkurl {MR3916109}}.
   {}
   \bibitem{MR4164728}
   J. Alt, L. Erd\H{o}s, and T. Kr\"{u}ger, \emph{The {D}yson equation with linear self-energy: spectral bands, edges
      and cusps}, Doc. Math. \textbf{25}, 1421–1539 (2020), \href {http://www.ams.org/mathscinet-getitem?mr=4164728} {\nolinkurl {MR4164728}}. {}
   \bibitem{MR3351052}
   F. Bekerman, A. Figalli, and A. Guionnet, \emph{Transport maps for {$\beta $}-matrix models and universality}, \href {https://doi.org/10.1007/s00220-015-2384-y} {Comm. Math. Phys. \textbf{338}, 589–619 (2015)}, \href {http://www.ams.org/mathscinet-getitem?mr=3351052} {\nolinkurl {MR3351052}}.
   {}
   \bibitem{MR3112927}
   G. Ben Arous and P. Bourgade, \emph{Extreme gaps between eigenvalues of random matrices}, \href {https://doi.org/10.1214/11-AOP710} {Ann. Probab. \textbf{41}, 2648–2681 (2013)}, \href {http://www.ams.org/mathscinet-getitem?mr=3112927} {\nolinkurl {MR3112927}}.
   {}
   \bibitem{MR4416591}
   P. Bourgade, \emph{Extreme gaps between eigenvalues of {W}igner matrices}, \href {https://doi.org/10.4171/jems/1141} {J. Eur. Math. Soc. (JEMS) \textbf{24}, 2823–2873 (2022)}, \href {http://www.ams.org/mathscinet-getitem?mr=4416591} {\nolinkurl {MR4416591}}.
   {}
   \bibitem{MR3253704}
   P. Bourgade, L. Erd\H{o}s, and H.-T. Yau, \emph{Edge universality of beta ensembles}, \href {https://doi.org/10.1007/s00220-014-2120-z} {Comm. Math. Phys. \textbf{332}, 261–353 (2014)}, \href {http://www.ams.org/mathscinet-getitem?mr=3253704} {\nolinkurl {MR3253704}}.
   {}
   \bibitem{MR3192527}
   P. Bourgade, L. Erd\H{o}s, and H.-T. Yau, \emph{Universality of general {$\beta $}-ensembles}, \href {https://doi.org/10.1215/00127094-2649752} {Duke Math. J. \textbf{163}, 1127–1190 (2014)}, \href {http://www.ams.org/mathscinet-getitem?mr=3192527} {\nolinkurl {MR3192527}}.
   {}
   \bibitem{MR3541852}
   P. Bourgade, L. Erd\H{o}s, H.-T. Yau, and J. Yin, \emph{Fixed energy universality for generalized {W}igner matrices}, \href {https://doi.org/10.1002/cpa.21624} {Comm. Pure Appl. Math. \textbf{69}, 1815–1881 (2016)}, \href {http://www.ams.org/mathscinet-getitem?mr=3541852} {\nolinkurl {MR3541852}}.
   {}
   \bibitem{MR4009717}
   Z. Che and B. Landon, \emph{Local spectral statistics of the addition of random matrices}, \href {https://doi.org/10.1007/s00440-019-00932-2} {Probab. Theory Related Fields \textbf{175}, 579–654 (2019)}, \href {http://www.ams.org/mathscinet-getitem?mr=4009717} {\nolinkurl {MR4009717}}.
   {}
   \bibitem{MR4026551}
   G. Cipolloni, L. Erd\H{o}s, T. Kr\"{u}ger, and D. Schr\"{o}der, \emph{Cusp universality for random matrices, {II}: {T}he real symmetric case}, \href {https://doi.org/10.2140/paa.2019.1.615} {Pure Appl. Anal. \textbf{1}, 615–707 (2019)}, \href {http://www.ams.org/mathscinet-getitem?mr=4026551} {\nolinkurl {MR4026551}}.
   {}
   \bibitem{1912.04100}
   G. Cipolloni, L. Erd\H{o}s, and D. Schr\"{o}der, \emph{Central Limit Theorem for Linear Eigenvalue Statistics of non-{H}ermitian Random Matrices}, \href {https://doi.org/10.1002/cpa.22028} {Comm. Pure Appl. Math. (2021)}, \href {https://arxiv.org/abs/1912.04100} {\nolinkurl {arXiv:1912.04100}}.
   {}
   \bibitem{MR4221653}
   G. Cipolloni, L. Erd\H{o}s, and D. Schr\"{o}der, \emph{Edge universality for non-{H}ermitian random matrices}, \href {https://doi.org/10.1007/s00440-020-01003-7} {Probab. Theory Related Fields \textbf{179}, 1–28 (2021)}, \href {http://www.ams.org/mathscinet-getitem?mr=4221653} {\nolinkurl {MR4221653}}.
   {}
   \bibitem{MR4235475}
   G. Cipolloni, L. Erd\H{o}s, and D. Schr\"{o}der, \emph{Fluctuation around the circular law for random matrices with real entries}, Electron. J. Probab. \textbf{26}, Paper No. 24, 61 (2021), \href {http://www.ams.org/mathscinet-getitem?mr=4235475} {\nolinkurl {MR4235475}}.
   {}
   \bibitem{MR2306224}
   P. Deift and D. Gioev, \emph{Universality at the edge of the spectrum for unitary, orthogonal, and symplectic ensembles of random matrices}, \href {https://doi.org/10.1002/cpa.20164} {Comm. Pure Appl. Math. \textbf{60}, 867–910 (2007)}, \href {http://www.ams.org/mathscinet-getitem?mr=2306224} {\nolinkurl {MR2306224}}.
   {}
   \bibitem{MR3068390}
   L. Erd\H{o}s, A. Knowles, H.-T. Yau, and J. Yin, \emph{The local semicircle law for a general class of random matrices}, \href {https://doi.org/10.1214/EJP.v18-2473} {Electron. J. Probab. \textbf{18}, no. 59, 58 (2013)}, \href {http://www.ams.org/mathscinet-getitem?mr=3068390} {\nolinkurl {MR3068390}}.
   {}
   \bibitem{MR4134946}
   L. Erd\H{o}s, T. Kr\"{u}ger, and D. Schr\"{o}der, \emph{Cusp universality for random matrices {I}: local law and the complex {H}ermitian case}, \href {https://doi.org/10.1007/s00220-019-03657-4} {Comm. Math. Phys. \textbf{378}, 1203–1278 (2020)}, \href {http://www.ams.org/mathscinet-getitem?mr=4134946} {\nolinkurl {MR4134946}}.
   {}
   \bibitem{MR3941370}
   L. Erd\H{o}s, T. Kr\"{u}ger, and D. Schr\"{o}der, \emph{Random matrices with slow correlation decay}, \href
   {https://doi.org/10.1017/fms.2019.2} {Forum Math. Sigma \textbf{7}, e8, 89 (2019)}, \href {http://www.ams.org/mathscinet-getitem?mr=3941370} {\nolinkurl {MR3941370}}.
   {}
   \bibitem{MR2810797}
   L. Erd\H{o}s, B. Schlein, and H.-T. Yau, \emph{Universality of random matrices and local relaxation flow}, \href {https://doi.org/10.1007/s00222-010-0302-7} {Invent. Math. \textbf{185}, 75–119 (2011)}, \href {http://www.ams.org/mathscinet-getitem?mr=2810797} {\nolinkurl {MR2810797}}.
   {}
   \bibitem{MR3699468}
   L. Erd\H{o}s and H.-T. Yau, \emph{A dynamical approach to random matrix theory}, Vol. 28, Courant Lecture Notes in Mathematics (Courant Institute of Mathematical Sciences, New York; American Mathematical Society, Providence, RI, 2017), pp. ix+226, \href {http://www.ams.org/mathscinet-getitem?mr=3699468} {\nolinkurl {MR3699468}}.
   {}
   \bibitem{MR3372074}
   L. Erd\H{o}s and H.-T. Yau, \emph{Gap universality of generalized {W}igner and {$\beta $}-ensembles}, \href
   {https://doi.org/10.4171/JEMS/548} {J. Eur. Math. Soc. ( JEMS) \textbf{17}, 1927–2036 (2015)}, \href {http://www.ams.org/mathscinet-getitem?mr=3372074} {\nolinkurl {MR3372074}}.
   {}
   \bibitem{MR2871147}
   L. Erd\H{o}s, H.-T. Yau, and J. Yin, \emph{Rigidity of eigenvalues of generalized {W}igner matrices}, \href {https://doi.org/10.1016/j.aim.2011.12.010} {Adv. Math. \textbf{229}, 1435–1515 (2012)}, \href {http://www.ams.org/mathscinet-getitem?mr=2871147} {\nolinkurl {MR2871147}}.
   {}
   \bibitem{MR1808248}
   P. J. Forrester and N. S. Witte, \emph{Exact {W}igner surmise type evaluation of the spacing distribution in the bulk of the scaled random matrix ensembles}, \href {https://doi.org/10.1023/A:1011074616607} {Lett. Math. Phys. \textbf{53}, 195–200 (2000)}, \href {http://www.ams.org/mathscinet-getitem?mr=1808248} {\nolinkurl {MR1808248}}.
   {}
   \bibitem{GPSW}
   H. Gharibyan, C. Pattison, S. Shenker, and K. Wells, \emph{Work in preperation}, (2021). {}
   \bibitem{MR2376207}
   J. W. Helton, R. Rashidi Far, and R. Speicher, \emph{Operator-valued semicircular elements: solving a quadratic matrix equation with positivity constraints}, \href {https://doi.org/10.1093/imrn/rnm086} {Int. Math. Res. Not. IMRN, Art. ID rnm086, 15 (2007)}, \href {http://www.ams.org/mathscinet-getitem?mr=2376207} {\nolinkurl {MR2376207}}.
   {}
   \bibitem{MR573370}
   M. Jimbo, T. Miwa, Y. M\^{o}ri, and M. Sato, \emph{Density matrix of an impenetrable {B}ose gas and the fifth
         {P}ainleve\'e transcendent}, \href {https://doi.org/10.1016/0167-2789(80)90006-8} {Phys. D \textbf{1}, 80–158 (1980)}, \href {http://www.ams.org/mathscinet-getitem?mr=573370} {\nolinkurl {MR573370}}.
   {}
   \bibitem{MR1810949}
   K. Johansson, \emph{Universality of the local spacing distribution in certain ensembles of {H}ermitian {W}igner matrices}, \href {https://doi.org/10.1007/s002200000328} {Comm. Math. Phys. \textbf{215}, 683–705 (2001)}, \href {http://www.ams.org/mathscinet-getitem?mr=1810949} {\nolinkurl {MR1810949}}.
   {}
   \bibitem{MR1876169}
   O. Kallenberg, \emph{Foundations of modern probability}, Second Edition, Probability and its Applications (New York) (Springer-Verlag, New York, 2002), pp. xx+638, \href {http://www.ams.org/mathscinet-getitem?mr=1876169} {\nolinkurl {MR1876169}}.
   {}
   \bibitem{MR3914908}
   B. Landon, P. Sosoe, and H.-T. Yau, \emph{Fixed energy universality of {D}yson {B}rownian motion}, \href {https://doi.org/10.1016/j.aim.2019.02.010} {Adv. Math. \textbf{346}, 1137–1332 (2019)}, \href {http://www.ams.org/mathscinet-getitem?mr=3914908} {\nolinkurl {MR3914908}}.
   {}
   \bibitem{MR3687212}
   B. Landon and H.-T. Yau, \emph{Convergence of local statistics of {D}yson {B}rownian motion}, \href {https://doi.org/10.1007/s00220-017-2955-1} {Comm. Math. Phys. \textbf{355}, 949–1000 (2017)}, \href {http://www.ams.org/mathscinet-getitem?mr=3687212} {\nolinkurl {MR3687212}}.
   {}
   \bibitem{1712.03881}
   B. Landon and H.-T. Yau, \emph{Edge statistics of {D}yson {B}rownian motion}, preprint (2017), \href {https://arxiv.org/abs/1712.03881} {\nolinkurl {arXiv:1712.03881}}.
   {}
   \bibitem{MR3405746}
   J. O. Lee and K. Schnelli, \emph{Edge universality for deformed {W}igner matrices}, \href {https://doi.org/10.1142/S0129055X1550018X} {Rev. Math. Phys. \textbf{27}, 1550018, 94 (2015)}, \href {http://www.ams.org/mathscinet-getitem?mr=3405746} {\nolinkurl {MR3405746}}.
   {}
   \bibitem{MR3502606}
   J. O. Lee, K. Schnelli, B. Stetler, and H.-T. Yau, \emph{Bulk universality for deformed {W}igner matrices}, \href {https://doi.org/10.1214/15-AOP1023} {Ann. Probab. \textbf{44}, 2349–2425 (2016)}, \href {http://www.ams.org/mathscinet-getitem?mr=3502606} {\nolinkurl {MR3502606}}.
   {}
   \bibitem{MR0220494}
   M. L. Mehta, \emph{Random matrices and the statistical theory of energy levels} (Academic Press, New York-London, 1967), pp. x+259, \href {http://www.ams.org/mathscinet-getitem?mr=0220494} {\nolinkurl {MR0220494}}.
   {}
   \bibitem{MR2012268}
   L. Pastur and M. Shcherbina, \emph{On the edge universality of the local eigenvalue statistics of matrix models}, Mat. Fiz. Anal. Geom. \textbf{10}, 335–365 (2003), \href {http://www.ams.org/mathscinet-getitem?mr=2012268} {\nolinkurl {MR2012268}}.
   {}
   \bibitem{MR3390602}
   M. Shcherbina, \emph{Change of variables as a method to study general {$\beta $}-models: bulk universality}, \href {https://doi.org/10.1063/1.4870603} {J. Math. Phys. \textbf{55}, 043504, 23 (2014)}, \href {http://www.ams.org/mathscinet-getitem?mr=3390602} {\nolinkurl {MR3390602}}.
   {}
   \bibitem{MR1727234}
   A. Soshnikov, \emph{Universality at the edge of the spectrum in {W}igner random matrices}, \href {https://doi.org/10.1007/s002200050743} {Comm. Math. Phys. \textbf{207}, 697–733 (1999)}, \href {http://www.ams.org/mathscinet-getitem?mr=1727234} {\nolinkurl {MR1727234}}.
   {}
   \bibitem{MR2784665}
   T. Tao and V. Vu, \emph{Random matrices: universality of local eigenvalue statistics}, \href {https://doi.org/10.1007/s11511-011-0061-3} {Acta Math. \textbf{206}, 127–204 (2011)}, \href {http://www.ams.org/mathscinet-getitem?mr=2784665} {\nolinkurl {MR2784665}}.
   {}
   \bibitem{MR3306005}
   T. Tao and V. Vu, \emph{Random matrices: universality of local spectral statistics of non-{H}ermitian matrices}, \href {https://doi.org/10.1214/13-AOP876} {Ann. Probab. \textbf{43}, 782–874 (2015)}, \href {http://www.ams.org/mathscinet-getitem?mr=3306005} {\nolinkurl {MR3306005}}.
\end{thebibliography}
\end{document}